\documentclass[journal]{IEEEtran}

\IEEEoverridecommandlockouts

\usepackage{mathrsfs}
\usepackage{latexsym}
\usepackage{graphicx}
\usepackage{epsfig}
\usepackage{subfigure}
\usepackage{array}
\usepackage{amsmath}
\usepackage{amssymb}
\usepackage{color, soul}
\usepackage{amsthm}
\usepackage{enumerate}
\usepackage{algpseudocode}
\usepackage{algorithm}
\usepackage{bm}
\usepackage{amssymb}
\usepackage{balance}
\usepackage{verbatim}

\theoremstyle{plain} 

\newtheorem{lemma}{Lemma}
\newtheorem{prop}{Proposition}
\theoremstyle{definition}

\def\bal#1\eal{\begin{align}#1\end{align}}

\newcommand{\bxi} {\boldsymbol{\xi}}
\newcommand{\bW}{{\bf W}}
\newcommand{\bH}{{\bf H}}
\newcommand{\bJ}{{\bf J}}
\newcommand{\bR}{{\bf R}}
\newcommand{\bP}{{\bf P}}
\newcommand{\bY}{{\bf Y}}
\newcommand{\bK}{{\bf K}}
\newcommand{\bI}{{\bf I}}
\newcommand{\bU}{{\bf U}}
\newcommand{\bV}{{\bf V}}
\newcommand{\bv}{{\bf v}}

\newcommand{\bN}{{\bf N}}
\newcommand{\bt}{{\bf t}}
\newcommand{\bA}{{\bf A}}

\newcommand{\bB}{{\bf B}}
\newcommand{\bC}{{\bf C}}
\newcommand{\bD}{{\bf D}}
\newcommand{\bQ}{{\bf Q}}
\newcommand{\bL}{{\bf L}}
\newcommand{\bZ}{{\bf Z}}

\newcommand{\bq}{{\bf q}}
\newcommand{\bT}{{\bf T}}
\newcommand{\bn}{{\bf n}}
\newcommand{\bX}{{\bf X}}

\newcommand{\bSigma} {\boldsymbol{\Sigma}}

\newcommand{\gl}{\lambda}

\newcommand{\by}{{\bf y}}
\newcommand{\bx}{{\bf x}}

\newcommand{\ba}{{\bf a}}

\newcommand{\br}{{\bf r}}
\newcommand{\bz}{{\bf z}}
\newcommand{\bs}{{\bf s}}
\newcommand{\bo}{{\bf 0}}

\newcommand{\bRequire}{{\bf Require}}

\newcommand{\bp} {\begin{proof}}
\newcommand{\ep} {\end{proof}}

\newcommand{{\Rb}} {\right)}

\newcommand{{\Rf}} {\right\}}
\newcommand{{\diag}} {\mathrm{diag}}

\begin{document}

\title{Enhancing Secure MIMO Transmission via Intelligent Reflecting Surface}

\author{Limeng Dong, Hui-Ming Wang \emph{Senior Member, IEEE}

\thanks{The authors are with the  School of Information and Communications Engineering, and also with the Ministry of Education Key Laboratory for Intelligent Networks and Network Security, Xi'an Jiaotong University, Xi'an 710049, China (e-mail: dlm$\_$nwpu@hotmail.com; xjbswhm@gmail.com)}

}

\maketitle

\begin{abstract}

 In this paper,  we consider an intelligent reflecting surface (IRS) assisted Guassian multiple-input multiple-output (MIMO) wiretap channel (WTC), and  focus on  enhancing its secrecy rate. Due to MIMO setting, all the existing solutions for enhancing the secrecy rate over multiple-input single-output WTC completely fall to this work. Furthermore, all the existing studies are simply based on an ideal assumption that full channel state information (CSI) of eavesdropper (Ev) is available. Therefore, we propose numerical solutions to enhance the secrecy rate of this channel under both full  and no Ev's CSI cases.
 For the full CSI case, we propose a barrier method and one-by-one (OBO) optimization combined alternating optimization (AO) algorithm to jointly optimize the transmit covariance $\bR$ at transmitter (Tx) and phase shift coefficient $\bQ$ at IRS. 
For the case of no Ev's CSI, we develop an artificial noise (AN) aided joint transmission scheme  to enhance the secrecy rate. In this scheme, a bisection search (BS) and OBO optimization combined AO algorithm is proposed to jointly optimize $\bR$ and $\bQ$. Such scheme is also applied  to enhance the secrecy rate under a special scenario in which the direct  link between Tx and receiver/Ev is blocked due to obstacles. In particular, we propose a BS and minorization-maximization (MM) combined AO algorithm with slightly faster convergence to optimize $\bR$ and $\bQ$ for this scenario.   Simulation results  have validated the monotonic convergence of the proposed algorithms, and it is shown that the proposed  algorithms  for the IRS-assisted design achieve significantly larger secrecy rate than the other benchmark schemes  under full CSI.  When Ev's CSI is unknown,  the secrecy performance of this channel also can be  enhanced by the proposed AN aided  scheme, and there is a trade-off between increasing the quality of service at Rx and enhancing the secrecy rate.

\end{abstract}


\section{Introduction}

Physical-layer security (PLS)  has emerged as a very valuable technology to deal with eavesdropping attacks in wireless systems.   In this approach, the
secrecy of communication is ensured at the physical layer
by exploiting the properties of wireless communication channels so that the transmitted information to users can be
completely “hidden” from eavesdropping and  cannot be recovered by
malicious eavesdroppers \cite{Bloch-11}.  This approach greatly  makes up for the defects of high complexity and high cost of hardware resources in traditional encryption method. Secrecy rate (capacity)  is  the key issue of guaranteeing the user's secret communication in PLS, and how to maximize the  secrecy rate  in multi-antenna wiretap channel (WTC) has drawn wide attentions in the past decade. The earliest research starts from the secrecy capacity of Gaussian multiple-input single-output (MISO) WTC \cite{Khisti-1}, and  several signal processing strategies such as artificial noise (AN) was also established for maximizing the secrecy rate \cite{Li-13}. Later on, The secrecy capacity of Gaussian multi-input multi-output (MIMO) WTC was deeply analyzed in \cite{Khisti-2}, and numerical solutions were proposed to maximize the secrecy rate \cite{Loyka-15}. Besides, the study of cognitive radio MIMO WTC was also established, and several analytical solutions were  being put forward to enhance user's secrecy performance \cite{Dong-18}\cite{Loyka-18}. In addition, large number of research results were established   to enhance the secrecy performance of 5G based multi-antenna wireless networks  via PLS \cite{Wu-18c}.

Recently, intelligent reflecting surface (IRS), also known as reconfigurable intelligent surface, has been proposed and it has drawn wide attention for its applications in wireless communications. IRS is a  software-controlled metasurface consisting of large numbers of passive reflecting elements. These reflecting elements could induce certain phase shift by a software based controller for the incident electromagnetic signal waves with very low power consumption \cite{Zhao-19} so that the propagation channel can be adjusted intelligently. Compared with the traditional reflecting surface, relaying and backscatter communications, the great benefits of IRS are concluded as three key aspects. Firstly, the phase shift in traditional reflecting surface  is fixed and cannot be changed, while IRS could continuously change the phase shift  by its small scale controller \cite{Hu-18}. Secondly, IRS is with low complexity and can be easily deployed on buildings, ceilings or indoor spaces, and it  is not affected by the receiver noise since it is not equipped with any signal processing equipments such as analog-to-digital, digital-to-analog converter and modulator or demodulator. Thirdly, IRS can be as a special ``relay" since it just reflects the  signal passively without any transmit power consumptions while the traditional relay requires a certain amount of  power for  signal transmission \cite{Ntontin-19}. These significant advantages make IRS as a green energy-efficient technique for beyond 5G and even 6G, and can be applied into various communication scenarios such as multi-cell, massive device-to-device, wireless information and power transfer, and secure communications  \cite{Wu-19}.

Several contributions have been established for boosting user's transmission rate by deploying IRS in wireless systems \cite{Wu-19b}-\cite{Zhang-19}. In \cite{Wu-19b}, an IRS-assisted MISO channel was considered and an algorithm was proposed to jointly optimize the active beamforming at transmitter and  passive phase shift coefficients at IRS. Simulation results showed that  a significantly larger  transmission rate can be achieved with the aid of IRS than that without IRS. Later, the model was extended to IRS-assisted MISO downlink multi-user channel in \cite{Huang-18}. The  energy efficiency  of IRS-assisted MISO downlink channel was also studied in \cite{Huang-19}\cite{Huang-18b}, and it was shown that with the aid of IRS, a huge increase of 300\% energy efficiency can be achieved compared with the regular multi-antenna amplify-and-forward relaying \cite{Huang-19}.  The algorithms for maximizing the transmission rate of IRS-assisted MIMO channel is established in \cite{Zhang-19}. All these results indicate that IRS greatly   boosts the transmission rate compared with no IRS case.

\subsection{Related work}

Inspired by these research results, IRS was recently combined with PLS to deal with eavesdropping attack issues. By adjusting the phase shift coefficients, the propagation channel between transmitter and receiver/eavesdropper can be adjusted so that the reflected signal by IRS not only can be added constructively with the non-reflected signal at the user, but also added destructively with the non-reflected signal at eavesdropper, thus significantly boosting the secrecy rate. Several research results about secure IRS-assisted MISO WTC  were established  \cite{Shen-19}-\cite{Xu-19}. In  \cite{Shen-19}-\cite{Cui-19}, the algorithms for secrecy rate maximization were proposed, and it was shown that IRS significantly helps improve the secrecy performance compared with no IRS  case. In \cite{Song-20}, a low complexity deep learning based solution was proposed, and it was shown that a comparable  secrecy performance can be achieved compared with the solution illustrated in \cite{Shen-19}. In \cite{Zheng-20}, it was shown that with IRS, the transmitter could use significantly less power to meet  the target secrecy rate  at receiver. A special case where the direct link between transmitter and receiver/eavesdropper was blocked was also considered in \cite{Yu-19}\cite{Feng-19}. The secrecy rate optimization of  multi-user MISO downlink WTC has been studied and several algorithms were proposed to maximize the secrecy rate \cite{Chen-19}\cite{Xu-19}. In addition, an algorithm for enhancing IRS-assisted MIMO WTC was studied in \cite{Dong-20}.
All these results again validate that IRS significantly enhance the user's secrecy rate compared with the existing solutions for no IRS case.

However, all the aforementioned works in the current literatures \cite{Shen-19}-\cite{Xu-19} are only restricted to MISO  setting, i.e., only one antenna at the receiver is considered. When  MIMO setting is considered, there are two significant differences about the
optimization problems compared with conventional MISO case.
Firstly, in MIMO systems, beamforming is not always optimal solution and hence the variable at transmitter changes from a beamforming vector to a covariance matrix   in the corresponding secrecy rate
maximization problem. Secondly, in the MIMO case, the objective
function in the problem is a complicated log of determinant expression compared with a simpler log of scalar formular in the MISO case. Therefore, the difficulty of maximizing the secrecy rate is significantly increased and all these existing numerical solutions for the MISO case fail to the MIMO case. Although \cite{Dong-20} proposes a numerical solution to enhance the IRS-assisted MIMO WTC, it only applies for a special case in which there is no direct communication link between transmitter and receiver/eavesdropper. When the direct links exist in the general case, how to numerically enhance the secrecy rate for this model is still an open problem. Furthermore, all these current works \cite{Shen-19}-\cite{Dong-20} are simply based on an ideal assumption that full CSI is available at transmitter, which is not practical since the eavesdropper is usually  a passive user and it does not actively exchange its CSI with the transmitter. And currently, there is no approaches about how to effectively guarantee secure communication in the IRS-assisted design if eavesdropper's CSI is unknown.

\subsection{Contributions}

Hence, motivated by these aforementioned significant benefits brought by IRS as well as the current research defects, in this paper,  we proceed to combine IRS with PLS  to 
enhance the secrecy performance of MIMO  channels. Specifically, we consider a general IRS-assisted Gaussian MIMO WTC in which one multi-antenna transmitter, receiver, eavesdropper as well as an IRS are involved, and  aim at developing numerical solutions to maximize its secrecy rate.  The main novelty and contribution of this paper is summarized as follows.

1). Firstly, we assume that full CSI is available at the transmitter, and to maximize the secrecy rate, an alternating optimization (AO) algorithm  is proposed to jointly optimize the transmit covariance $\bR$ at transmitter as well as phase shift coefficient $\bQ$ at IRS in two independent sub-problems.  To optimize $\bR$ given $\bQ$, the non-convex sub-problem is firstly equivalently transformed to a convex-concave problem whose optimal solution is a saddle point. Then,  a barrier method in combination with Newton method and backtracking line search method is proposed to globally optimize $\bR$. To optimize $\bQ$ given $\bR$ in the non-convex sub-problem, an one-by-one (OBO) optimization method is proposed in which each one of the $n$ phase shift coefficients are optimized in order by fixing the other $n-1$ coefficients as constant. As the convergence is reached, the results returned by the AO is a limit point solution of the original problem.  Simulation results show that our algorithm for the proposed IRS-assisted design  greatly enhance the secrecy rate compared with the existing benchmark schemes with and without IRS.

2). Secondly, we assume that the eavesdropper's CSI is completely unknown at transmitter. To maximize the secrecy rate of this channel given a fixed total power at transmitter, we propose  an AN aided joint transmission scheme, in which a minimum transmit power is firstly optimized subject to a quality-of-service (QoS) constraint by jointly optimizing $\bR$ and $\bQ$,  and then AN is applied to jam the eavesdropper by using the residual power at transmitter. When solving the power minimization problem, a bisection search (BS) and OBO combined AO algorithm is to jointly optimize $\bR$ given $\bQ$. As the convergence is reached, the results returned by the AO is also a limit point solution of the original problem. Simulation results show that our proposed AN aided joint transmission scheme also greatly enhance the secrecy rate under  QoS constraint, and it is shown that there is a trade-off between increasing the QoS  and enhancing secrecy rate.

3). For the case of no eavesdropper's CSI, a special scenario is considered in which the direct communication link between the transmitter and receiver/eavesdropper is blocked by obstacles, and  AN aided  scheme is still  applied to enhance the secrecy rate in this case. In particular, in addition to BS and OBO combined   AO algorithm, we propose a BS and minorization-maximization (MM) combined algorithm to  solve the power minimization problem. In this MM algorithm, all the $n$ phase shift coefficients are simultaneously optimized iteratively given fixed $\bR$ in the sub-problem. The key difficulty is how to obtain a proper lower bound (i.e., surrogate function) of the objective function in the sub-problem so that MM can be applied to optimize $\bQ$. Therefore, we propose three successive approximations for the objective function to find a proper surrogate lower bound of the objective function, which  is significantly different from the MISO case \cite{Shen-19}\cite{Yu-19} where only one approximation is used to obtain the bound due to the simple structure of the objective function.  Simulation results show that the proposed MM and BS combined AO algorithm has less performance on enhancing the secrecy rate but with slightly faster speed of convergence  compared with the  OBO and BS combined AO algorithm.

The rest of the paper is organized as follows: Section II describes the channel model and formulate the optimization problem. In section III, the AO algorithm is proposed to jointly maximize $\bR$ and $\bQ$ under full CSI case. In section IV, the AN aided joint transmission scheme is proposed to maximize the secrecy rate under no eavesdropper's CSI case. Simulation results have been carried out to evaluate the performance and convergence of proposed algorithm in section V. Finally, section VI concludes the paper.

\emph{Notations}: bold lower-case letters ($\ba$) and capitals ($\bA$) denote vectors and matrices respectively; $\bA^{\rm T}$, $\bA^{*}$ and $\bA^{\rm H}$ denote transpose, conjugate and  Hermitian conjugate of $\bA$, respectively; $\bA \geq \bo$ means  positive semi-definite; ${\rm E}\left \{ \cdot  \right \}$ is statistical expectation, $\gl_i(\bA)$ denotes eigenvalues of $\bA$, which are in decreasing order unless indicated otherwise, i.e. $\gl_1\ge\gl_2\ge\gl_3...$; ${\rm rank}(\bA)$ denotes the rank of $\bA$; $|\bA|$ and ${\rm tr}(\bA)$ are determinant and trace of $\bA$; $\bI$ is an identity matrix of appropriate size; $|\ba|$ denotes the norm of the vector $\ba$; $\mathbb{C}^{M \times N}$ and $\mathbb{R}^{M \times N}$ denote the space of $M\times N$ matrix with complex-valued elements and real-valued elements, respectively; $\otimes$ denotes Kronecker product and $\odot$ denotes Hadamard product; ${\rm vec}(\bA)$ is the vector obtained by stacking all columns of matrix $\bA$ on top of each other; ${\rm arg}(a)$ denotes the phase of the complex value $a$;  ${\rm diag}(\ba)$ is to transform the vector $\ba$ as a diagonal matrix with diagonal elements in $\ba$; ${\rm Re}\{\bA\}$ denotes the real elements in $\bA$.

\section{Channel Model And Problem Formulation}

Let us consider an IRS-assisted MIMO WTC model shown as Fig.1, in which a  transmitter Alice, receiver Bob, eavesdropper Eve and an IRS are included. The  number of antennas deployed at Alice, Bob and Eve are $m$, $d$, $e$ respectively, and the number of reflecting elements on the IRS is $n$. The task for IRS in this model is to adjust the phase shift
coefficient of the reflecting elements by the controller, and reflect the  information signals from
Alice passively to Bob and Eve (without generating any extra noise) so as to constructively add
with the non-reflected signal from Alice-Bob link and destructively add with the non-reflected
signal from Alice-Eve link.   Based on this setting, the received signals at Bob and Eve  are expressed as
\bal
\notag
&\by_{B} = \bH_{AB}\bx+\bH_{IB}\bQ\bH_{AI}\bx+\bxi_{B},\\
\notag
  &\by_{E}=\bH_{AE}\bx+\bH_{IE}\bQ\bH_{AI}\bx+\bxi_{E}
\eal
respectively where $\bx\in\mathbb{C}^{m\times1}$ is the transmitted signal, $\bH_{AB}\in\mathbb{C}^{d\times m}$, $\bH_{AE}\in\mathbb{C}^{e\times m}$, $\bH_{AI}\in\mathbb{C}^{n\times m}$, $\bH_{IB}\in\mathbb{C}^{d\times n}$ and $\bH_{IE}\in\mathbb{C}^{e\times n}$  are the channel matrices representing the direct link of   Alice-Bob, Alice-Eve,  Alice-IRS, IRS-Bob and IRS-Eve respectively, $\bxi_{B}\in\mathbb{C}^{d\times 1}$ and $\bxi_{E}\in\mathbb{C}^{e\times 1}$ represent complex noise at  Bob and Eve respectively with i.i.d entires distributed as $\mathcal{CN}(0,1)$,  
$\bQ={\rm diag}([q_1,q_2,...,q_n]^{\rm T})$ is the diagonal phase shift matrix for IRS, $q_i=e^{j\theta_i}$ is the phase shift coefficient at reflecting element $i$ ($i=1,2,...,n$). In addition, the controller shown in Fig.1 is used to coordinate Alice and IRS
for channel acquisition and data transmission tasks \cite{Zhang-19}.

 Based on the this signal model $\by_{B}$ and $\by_{E}$,  the achievable transmission rate $C_B$ at Bob and $C_E$ at Eve can be expressed as
\bal
\notag
&C_B\\
\notag
=&\log_2\left |\bI+(\bH_{AB}+\bH_{IB}\bQ\bH_{AI})\bR(\bH_{AB}+\bH_{IB}\bQ\bH_{AI})^{\rm H}\right |,\\
\notag
&C_E\\
\notag
=&\log_2\left |\bI+(\bH_{AE}+\bH_{IE}\bQ\bH_{AI})\bR(\bH_{AE}+\bH_{IE}\bQ\bH_{AI})^{\rm H}\right |
\eal
respectively
 where $\bR={\rm E}\{\bx\bx^{\rm H}\}$ is the transmit covariance matrix. Therefore, based on the key concept of information-theoretic PLS, to guarantee secure communication for this channel, the achievable secrecy rate $C_B-C_E>0$ should holds, and larger
secrecy rate indicates better secrecy performance \cite{Bloch-11}. In this paper, we will focus on enhancing the secrecy rate of this channel by jointly optimizing $\bR$ and $\bQ$ based on two conditions of CSI: full CSI\footnote{Note that for full CSI available at Alice, this can be achieved by modern adaptive
system design, where channels are estimated at Bob and Eve,
and send back to Alice. Since Eve is just other user in
the system and it also share its CSI with Alice but is untrusted
by Bob.  For how to estimate the channels, we apply the existing solutions (see e.g. \cite{Wang-19}-\cite{Mirza-19}) to obtain the direct link $\bH_{AB}$ and $\bH_{AE}$  as well as the reflecting link $\bH_{AI}$, $\bH_{IB}$ and $\bH_{IE}$.} and completely no Eve's CSI at Alice. 

\begin{figure}[t]
	\centerline{\includegraphics[width=3.0in]{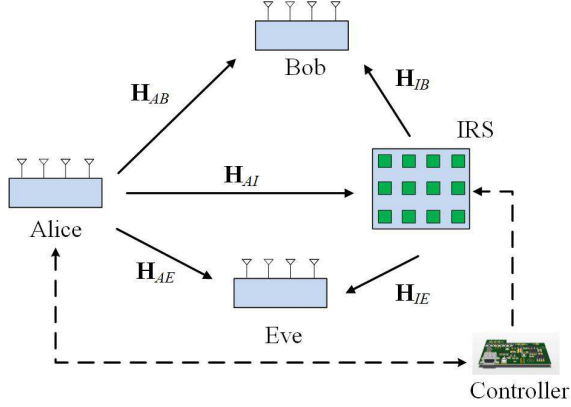}}
	\caption{A block diagram of IRS-assisted Gaussian MIMO WTC}
\end{figure}

\section{AO Algorithm for Enhancing Secrecy Rate  Under Full CSI}

In this section, we assume that full CSI is available at Alice, and focus on enhancing the secrecy rate of IRS-assisted MIMO WTC.  Based on the aforementioned system setting, the  secrecy rate optimization problem of this channel model is expressed as 
\bal
\label{origial_problem}
{\rm P1}: \underset{\bR\in S_{\bR}\bQ\in S_{\bQ}}{\max}\ C_s(\bR,\bQ)=C_B-C_E
\eal
where 
\bal
\notag
&S_{\bR}\triangleq\{\bR:\bR\geq \bo, {\rm tr}(\bR)\leq P\},\\
\notag
  &S_{\bQ}\triangleq\{\bQ: |q_i|=1, \forall i\}
\eal
denote the feasible set of the transmit covariance $\bR$ and  phase shift matrix $\bQ$, respectively, ${\rm tr}(\bR)\leq P$ is the total power constraint (TPC) at Alice, $P$ denotes  total transmit power budget,  the unit modulus constraint (UMC) $|q_i|=1$ ensures that each reflecting element in IRS does not change the amplitude of the signal. We note that this is a complicated non-convex problem due to non-convex objective function as well as non-convex constraint. Since MIMO setting is considered in this work, all the existing solutions for the MISO case \cite{Shen-19}-\cite{Xu-19}  completely fail to solve ${\rm P1}$ in the full MIMO setting. The reason is that   the determinant term in the objective function $C_B-C_E$ cannot be equivalently simplified to a scalar formular as in the MISO case. Although \cite{Dong-20} studied the secrecy rate enhancement of IRS-assisted MIMO WTC, its proposed solution still cannot be directly applied to our considered system model in which the direct link  channels $\bH_{AB}$ and $\bH_{AE}$ exist  in addition to the reflecting link channels $\bH_{AI}$, $\bH_{IB}$ and $\bH_{IE}$. 

Hence, in this section, we develop a  numerical solution to solve the non-convex ${\rm P1}$, which is based on AO algorithm to jointly optimize  $\bR$ and $\bQ$. The main reason for choosing this  algorithm is due to two aspects. Firstly,  different from the secrecy capacity of MIMO WTC optimization problem \cite{Khisti-2}-\cite{Loyka-18} in which only one variable $\bR$ is considered, there are two variables $\bR$ and $\bQ$ need to be optimized, also note that the objective function $C_B-C_E$ is a non-convex complicated form and it is difficult to directly obtain the optimal solutions. Hence, the function  of AO algorithm is to split the non-convex problem ${\rm P1}$ into two sub-problem with simpler structure by fixing each variable as a constant so that  $\bR$,  $\bQ$ can be optimized separately in each sub-problem. Secondly, although $\bR$ and $\bQ$ are bounded by the TPC and UMC respectively, these two constraints are independent between each other. Therefore, using AO algorithm, a monotonic convergence of the objective value  can be achieved so that a limit point solution of  ${\rm P1}$ can be obtained. In our AO algorithm,  two new solutions were developed  to optimize $\bR$ and $\bQ$  for each sub-problem in the following subsections. When optimizing $\bR$ given fixed $\bQ$, the non-convex sub-problem is firstly equivalently transformed to a convex-concave problem, and a barrier method in combination with Newton method and backtracking line search method is proposed to  globally optimize $\bR$. When optimizing $\bQ$ given fixed $\bR$, we propose an OBO optimization method  to optimize a sub-optimal solution of $\bQ$. 

\subsection{Barrier Method for Optimizing Transmit Covariance at  Alice}

In this subsection, we fix phase shift $\bQ$ as a constant and optimize the covariance $\bR$. When  $\bQ$ is fixed, the sub-problem of optimizing $\bR$ is expressed as
\bal
{\rm P2}: \underset{\bR\in S_{\bR}}{\max} \ C(\bR)=\log_2\frac{|\bI+\bH_1\bR\bH_1^{\rm H}|}{|\bI+\bH_2\bR\bH_2^{\rm H}|}
\eal
where $\bH_1=\bH_{AB}+\bH_{IB}\bQ\bH_{AI}, \ \ \bH_2=\bH_{AE}+\bH_{IE}\bQ\bH_{AI}$.  
It can be known that this is a standard secrecy capacity optimization problem of Gaussian MIMO WTC. This is  a difficult non-convex problem, however, it can be equivalently transformed to a convex-concave formular. By applying the key theorem in \cite{Khisti-2},  ${\rm P2}$ can be equivalently expressed as 
\bal
 {\rm P3}:\ \max_{\bR\in S_{\bR}} \min_{\bK\in S_{\bK}} f(\bR,\bK)=\log_2\frac{|\bI+\bK^{-1}\bH\bR\bH^{\rm H}|}{|\bI+\bH_2\bR\bH_2^{\rm H}|}
\eal
where  $\bH=[\bH_1^{\rm H}, \bH_2^{\rm H}]^{\rm H}$, $\bK=E\{\bxi\bxi^{\rm H}\}\in\mathbb{C}^{2(d+e)\times 2(d+e)}$, $\bxi=[\bxi_{B1}^{\rm H}, \bxi_{B2}^{\rm H}, \bxi_{E1}^{\rm H}, \bxi_{E2}^{\rm H}]^{\rm H}$, $S_{\bK}$ is the feasible set of $\bK$ defined as
\bal
S_{\bK}\triangleq \left \{\bK:\bK=
\begin{bmatrix}
\bI & \bN^{\rm H} \\
\bN & \bI
\end{bmatrix},\ \bK \ge \bo \right \}
\eal
where $\bN=E\{\bxi_E\bxi_B^{\rm H}\}\in\mathbb{C}^{2e\times 2d}$, $\bxi_B=[\bxi_{B1}^{\rm H}, \bxi_{B2}^{\rm H}]^{\rm H}$, $\bxi_E=[\bxi_{E1}^{\rm H}, \bxi_{E2}^{\rm H}]^{\rm H}$.
The theorem in \cite{Khisti-2} indicates that ${\rm P3}$ is a convex-concave optimization problem with saddle point solution, i.e., the objective is concave in $\bR$ for any given $\bK$ and convex in $\bK$ for any given $\bR$ so that Karush-Kuhn-Tucker (KKT) is the sufficient and necessary conditions for optimality. Also the saddle point solution $\bR$ for the max-min problem ${\rm P3}$ is always the optimal solution for the original max problem ${\rm P2}$. Hence, the work left is how to numerically obtain the saddle point solution for ${\rm P3}$. In this subsection, we apply the barrier method illustrated in \cite{Loyka-15} to solve ${\rm P3}$, which is in combination with Newton method and backtracking line search method. Note that  in \cite{Loyka-15}, this method was developed only based on  real-valued channel matrix case. In this paper, we improve this algorithm  by re-deriving the gradients and Hessians of the barrier function so that it can be used under general complex-valued channel matrix case.

Specifically, by introducing a barrier parameter $t>0$, the constraints can be absorbed by the objective function in P3 so that the barrier function can be expressed as
\bal
\notag
 f_t(\bR,\bK)=&f(\bR,\bK)+t^{-1}\log_2(P-tr(\bR))\\
\notag
 &+t^{-1}\log_2|\bR|-t^{-1}\log_2|\bK|.
\eal
Hence, a new optimization problem can be formulated as 
\bal
{\rm P4}: \ \max_{\bR\in S'_{\bR}} \min_{\bK\in S'_{\bK}} \ f_t(\bR,\bK)
\eal
where $S'_{\bR}=\{\bR:\bR> \bo, {\rm tr}(\bR)< P\}, \ S'_{\bK}=\{\bK:\bK> \bo\}$. Since $S'_{\bR}\in S_{\bR}$ and $S'_{\bK}\in S_{\bK}$, $\rm P4$ is still a convex-concave optimization problem so that KKT conditions is still sufficient and necessary for optimality. Thus, the work reduces to find the saddle point satisfying the  KKT conditions $\nabla_\bR f_t(\bR,\bK)=\bo,  \nabla_\bK f_t(\bR,\bK)=\bo$.
 Since the variable $\bR$ and $\bK$ are Hermitian matrices, it is difficult to obtain the gradient (KKT conditions) and Hessians of the objective function. Thus, a vectorization for the variables is needed. Note that in \cite{Loyka-15},  $\bR$ and $\bK$ are symmetric real matrix and its vectorization as well as the obtained gradients and Hessians  cannot be directly applied to our complex-valued case. In this paper, we make new vectorization for the  complex $\bR$ and $\bK$. Specifically, let 
$\br=[\br_d^{\rm T}, \br_l^{\rm T}, \br_l^{\rm H}]^{\rm T}, \bn=[{\rm vec}(\bN)^{\rm T}, {\rm vec}(\bN)^{\rm H})]^{\rm T}$, where $\br_d$ denotes vectorizing all the diagonal real elements of $\bR$ and  $\br_l$ denotes vectorizing all the lower triangular complex elements of $\bR$, then $\br$ and $\bn$ can be further expressed as the  linear transformation $\br=\bD_{\br}^{\rm T}{\rm vec}(\bR), \bn=\bD_{\bn}^{\rm T}{\rm vec}(\bK-\bI)$ where $\bD_{\br}\in\mathbb{R}^{m^2 \times m^2}$ and $\bD_{\bn}\in\mathbb{R}^{(d+e)^2 \times 2de}$ are unique full column rank matrices (with the elements either zero and one) satisfying $\bD_{\br}^{\rm T}\bD_{\br}=\bI$ and $\bD_{\bn}^{\rm T}\bD_{\bn}=\bI$. For the details about how to construct $\bD_{\br}$ and $\bD_{\bn}$, please refer to \cite{Hjorungnes-11}. Let $\bz=[\br^{\rm H}, \bn^{\rm H}]^{\rm H}$, since $\bz$ can represent all the key information of $\bR$ and $\bK$ completely, the work now  further reduces to find the optimal $\bz$ satisfying the new KKT conditions $r(\bz)=[(\nabla_{\br^{*}} f_t(\bR,\bK))^H \nabla_{\bn^{*}} f_t(\bR,\bK)^H]^H=\bo$.
In Newton method, the optimality condition $r(\bz)=\bo$ is iteratively solved using 1st-order approximation of $r(\bz)$, which corresponds to the 2nd order approximation of the objective function
\bal
\label{residual}
r(\bz_k+\Delta\bz)=r(\bz_k)+\bT\Delta\bz+o(\Delta\bz)=\bo
\eal
where $\bz_k$ and $\Delta\bz$ are the current variables and their updates at iteration $k$ respectively, and where $\bT$ is the derivative of $r(\bz)$, i.e., the Hessian matrix of $f_t(\bR,\bK)$ in $\br$ and $\bn$:
\bal
\bT=\begin{bmatrix}
 \nabla_{\br^{*}\br^{\rm T}}^2 f_t(\bR,\bK) & \nabla_{\br^{*}\bn^{\rm T}}^2 f_t(\bR,\bK)\\
[\nabla_{\br^{*}\bn^{\rm T}}^2 f_t(\bR,\bK)]^{\rm H} & \nabla_{\bn^{*}\bn^{\rm T}}^2 f_t(\bR,\bK)
\end{bmatrix}.
\eal
Closed-form expressions for gradients $\nabla_{\br^{*}} f_t(\bR,\bK)$, $\nabla_{\bn^{*}} f_t(\bR,\bK)$ and Hessians $\bT$ are given in the Appendix. By ignoring $o(\Delta\bz)$, \eqref{residual} can be further expressed as linear equation
\bal
\label{update}
r(\bz_k)+\bT\Delta\bz=\bo
\eal
so that the update $\Delta\bz$ can be solved numerically from this  equation using the existing  solver (such as ``linsolve" function in Matlab). After that, $\bz$ can be updated as  $\bz_{k+1}=\bz_k+s\Delta\bz$ where $s>0$ denotes the step size which can be found via backtracking line search.

The proposed algorithm for solving P2 is summarized as Algorithm 1. In this algorithm, $\alpha$ is the percentage of the liner decrease for the residual norm $|r(\bz_k)|$ in backtracking line search, $\beta$ controls the reduction in step size at each iteration of backtracking line search, $\eta$ controls the increase of $t$ at each iteration of barrier method, $\epsilon_1$ is the required computing accuracy, and $t\in[t_0, t_{max}]$ where $t_0$, $t_{max}$ are the initial and maximum value of $t$ respectively, $(\bR_{0}, \bK_{0})$ is the feasible starting point satisfying $\bR_0\in S'_{\bR}$ and $\bK_0\in S'_{\bK}$. Once the target accuracy of Newton method is reached, the  computed $(\bR(t),\bK(t))$ is set as the new starting point for the new problem P4 with updated $t$. As $t$ has reached the maximum $t_{max}$, the algorithm stops and output the final solution $\bR$. Note that $t_{max}$ should not be set too high, or the Hessian matrix $\bT$ will get close to singular so that the update $\Delta\bz$ cannot be computed in practice. Using the same steps of proof illustrated in \cite{Loyka-15}, it can be verified that the residual norm $|r(\bz_k)|$ is strictly decreasing in the Newton steps, and also Hessian matrix $\bT$ is non-singular for each $t>0$ so that the solution $\Delta\bz$ is unique during each iteration. Finally, by applying the key properties of barrier method \cite{Boyd-04}, one obtains that as $t\rightarrow \infty$, $f(\bR(t), \bK(t))\rightarrow C_{opt}$, from which Algorithm 1 is guaranteed to global convergence.

\begin{algorithm}[h]
	\caption{(\it solving ${\rm P2}$ given fixed $\bQ$)}
	\begin{algorithmic}
		\State \bRequire\  $(\bR_{0}, \bK_{0})\rightarrow \bz_0$, $0<\alpha <0.5$, $0<\beta<1$, $t_{0}>0$, $t_{max}>t_0$, $\mu>1$, $\epsilon_1 >0$.
        \State 1. Set $t=t_0$.
		\Repeat\ \ (start barrier method)
        \State 2. Set $k=0$.
		\Repeat\ \ (start Newton method)
		\State 3. Compute $r(\bz_{k})$  for current $k$, and compute update $\Delta \bz$ via \eqref{update}, and set $s=1$.
		\Repeat\ \ (start backtracking line search method)
		\State 4. $s:=\beta s$, update $\bz_{k+1}=\bz_{k}+s\Delta \bz$.
		\Until {$|r(\bz_{k+1})|\leqslant (1-\alpha s)|r(\bz_{k})|$ and  $\bR_{k+1} \in S'_{\bR},\bK_{k+1} \in S'_{\bK}$}
		\State 5. $k:=k+1$.
		\Until{$|r(\bz_{k})|\leqslant \epsilon_1$}
        \State 6. Set $\bz_{0}:=\bz_{k}$ as a new starting point, and update $t:=\mu t$.
		\Until {$t>  t_{max}$}
	\end{algorithmic}
\end{algorithm}

In Algorithm 1,  when $t$ is fixed, the main computational cost comes from the residual norm $r(\bz_k)$, the update $\Delta\bz$ and the loop of backtracking line search in each iteration of Newton method. The complexity of computing $r(\bz_k)$ and $\Delta\bz$ are  about $u_1=O(m^4+(d+e)^3+m(d+e)^2+em^2+2de(d+e)^2 )$ and $u_2=O((m^2+de)^3)$ respectively. The complexity of finding the step size $s$ in each iteration of backtracking line search is about $u_3=O(m^2+de)$. If $l_b$ is the total iterations required for backtracking line search to converge, one obtains that the total computation complexity for the  current iteration of Newton method  is $u_1+u_2+l_bu_3$. Furthermore, according to the property of standard barrier method \cite{Boyd-04}, the total number of Newton steps for each value of $t$ scales as $l_n=O(log_2\epsilon_1^{-1}\sqrt{\frac{m(m+1)}{2}+de})$. Thus, the total complexity of Algorithm 1 for each fixed $t$ is $O(l_n(u_1+u_2+l_bu_3))$.

\subsection{Algorithm of Optimizing Phase Shift Matrix at the IRS}

 With the algorithm of optimizing $\bR$ for fixed $\bQ$, the next step is to compute $\bQ$  in this sub-section. The sub-problem of optimizing $\bQ$ given fixed $\bR$ is expressed as ${\rm{P5}}$.
\bal
\label{origial_problem}
{\rm{P5}}: \underset{\bQ}{\max}\ C_s(\bR,\bQ),\ s.t.\  |q_i|=1, \forall i.
\eal
Since the objective function in ${\rm{P5}}$ is a complicated log determinant function, the existing solutions such as semi-definite relaxation and MM in \cite{Shen-19}\cite{Cui-19} fail to solve this problem. Hence, inspired by \cite{Zhang-19}, we propose an OBO optimization method  to optimize $\bQ$,  in which each one of the $n$ phase shift coefficients are optimized in order by fixing the other $n-1$ coefficients as constant. Note that the system model illustrated in \cite{Zhang-19} is only an IRS-assisted MIMO channel, and its solution  can not be directly applied to our WTC case. Hence, our OBO optimization method is based on the method in \cite{Zhang-19} but with proper extensions so that a  sub-optimal $\bQ$ for ${\rm{P5}}$ can be obtained. Specifically, consider the $i$-th phase shift coefficient is unknown and all the rest $n-1$ coefficients are given, and let the eigenvalue decomposition of $\bR$ is $\bR=\bU_{\bR}\bSigma_{\bR}\bU_{\bR}^{\rm H}$  where $\bU_{\bR}$ is the unitary matrix, the columns of which are the eigenvectors of $\bR$, $\bSigma_{\bR}$ is the diagonal matrix, in which the diagonal entries are eigenvalues of $\bR$. The following proposition concludes the closed-form optimal solutions of $q_i$.
\begin{prop}
Given fixed $q_1, q_2,..., q_{i-1}, q_{i+1},..,q_n$, then ${\rm P5}$ can be simplified to the following problem
${\rm P6}$
\bal
{\rm P6}: \underset{q_i}{\max}\ C'_B(q_i)-C'_E(q_i), \ \ s.t.\  |q_i|=1, \forall i
\eal
where $C'_B(q_i)=\log_2|\bI+q_i\bA_i^{-1}\bB_i+q_i^*\bA_i^{-1}\bB_i^{\rm H}|$, $C'_E(q_i)=\log_2|\bI+q_i\bC_i^{-1}\bD_i+q_i^*\bC_i^{-1}\bD_i^{\rm H}|$.
When ${\rm tr}(\bA_i^{-1}\bB_i)={\rm tr}(\bC_i^{-1}\bD_i)=0$, then  $q_i=e^{j\theta_i}$ where  $\theta_i$ can be any value between $0$ and $2\pi$. When ${\rm tr}(\bA_i^{-1}\bB_i)\neq0, {\rm tr}(\bC_i^{-1}\bD_i)=0$, then the  optimal solution of $q_i$ is expressed as
$q_i=e^{-jarg(\bar{\lambda}_i)}$, where $\bar{\lambda}_i$ is the only non-zero eigenvalue of $\bA_i^{-1}\bB_i$.
When ${\rm tr}(\bA_i^{-1}\bB_i)=0, {\rm tr}(\bC_i^{-1}\bD_i)\neq0$, then  the  optimal solution of $q_i$ is expressed as
$q_i=e^{j(\pi-arg(\tilde{\lambda}_i))}$, where $\tilde{\lambda}_i$ is the only non-zero eigenvalue of $\bC_i^{-1}\bD_i$.
When ${\rm tr}(\bA_i^{-1}\bB_i)\neq0, {\rm tr}(\bC_i^{-1}\bD_i)\neq0$, then the  optimal solution of $q_i$ is expressed as
$q_i=e^{-jarg(\bar{\lambda}_i-u\tilde{\lambda}_i)}$, 
where $u>0$ is found from BS algorithm.
\end{prop}

\begin{proof}
Detailed expressions of  $\bA_i, \bB_i, \bC_i, \bD_i$ as well as the proof are provided in Appendix.
\end{proof}

Based on Proposition 1, the OBO optimization algorithm for optimizing $\bQ$ given $\bR$ is summarized as Algorithm 2. In this algorithm, $q_i^{ini}$ is the initialized feasible starting point for each phase shift coefficient. Once all the coefficients are optimized in order, the output $\bQ$ returned by OBO algorithm is a sub-optimal solution of ${\rm P5}$. In this algorithm, the main computational  complexity comes from computing the eigenvalue of $\bA_i^{-1}\bB_i$ and $\bC_i^{-1}\bD_i$. Hence, if there are $n$ reflecting elements, the total complexity of Algorithm 2 is about $O(2n(d^3+e^3))$.

\begin{algorithm}[h]
	\caption{(\it OBO optimization algorithm for solving P5 given fixed $\bR$)}
	\begin{algorithmic}
		\State \bRequire\  $q_i^{ini}$, $i=1,2,...n$.
         \State 1.Set $i=1$.
		\Repeat\ \ 
         \State 2. Compute ${\rm tr}(\bA_i^{-1}\bB_i)$ and ${\rm tr}(\bC_i^{-1}\bD_i)$.
         \State 3.  Obtain the optimal solution of $q_i$  according to Proposition 3.
           \State 4. Set $i=i+1$.
		\Until{$i=n$} 
       \State 5. Output  $\bQ={\rm diag}\{[q_1,q_2,..., q_n]^{\rm T}\}$ as the sub-optimal solution of P5.  
	\end{algorithmic}
\end{algorithm}

\subsection{Summary of the AO Algorithm}
Finally, with Algorithm 1 and Algorithm 2 at hand, the AO algorithm for maximizing the secrecy rate of IRS-assisted MIMO WTC  is summarized as Algorithm 3. 
Since $\bR$ and $\bQ$ are optimized alternatively, the value of objective function $C_s(\bR,\bQ)$ in P1 is non-decreasing with iterations, i.e,
$C_s(\bR_{1},\bQ_{1})\leq C_s(\bR_{2},\bQ_{2})\leq...\leq C_s(\bR_{k},\bQ_{k})$,
where $\bR_k$ and $\bQ_k$ is the optimized solution of P2 and P5 returned by Algorithm 1 and Algorithm 2 respectively at iteration $k$. Furthermore,  since $\bR$ and $\bQ$ are both bounded by the feasible set $S_{\bR}$ and $S_{\bQ}$ respectively, by applying the Cauchy's theorem \cite{Chen-19}, one obtains that a solution $\bR_{opt}$ and $\bQ_{opt}$ always exist such that
\bal
\notag
0&=\underset{k\rightarrow \infty}\lim\{C_s(\bR_{k},\bQ_{k})-C_s(\bR_{opt},\bQ_{opt})\}\\
\notag
&\leq \underset{k\rightarrow \infty}\lim\{C_s(\bR_{k+1},\bQ_{k+1})-C_s(\bR_{opt},\bQ_{opt})\}=0,
\eal
which indicates that Algorithm 3 is guaranteed to converge to a limit point solution of P1.

\begin{algorithm}[h]
	\caption{(\it AO algorithm of solving {\rm P1})}
	\begin{algorithmic}
		\State \bRequire\  $\epsilon_2 >0$.
        \State 1. Set initial point $\bR_{ini}$, $\bQ_{ini}$.
        \State 2. Compute  $C_1=C_B-C_E$ given $\bR_{ini}$ and $\bQ_{ini}$.
		\Repeat\ \ 
     \State 3. Given fixed $\bR_{ini}$,  solve the  sub-optimal solution $\bQ_{opt}$ of P5 using Algorithm 2.
         \State 4. Solve the global optimal solution $\bR_{opt}$ of P2 given  $\bQ_{opt}$ using Algorithm 1.
          \State 5. Compute the current object value $C_2=C_B-C_E$ under $\bR_{opt}$ and $\bQ_{opt}$.
		  \State 6.  If $|C_2-C_1|/|C_1|>\epsilon_2$, set $\bQ_{opt}=\bQ_{ini}$ and $C_1=C_2$, go back to step 3.
		\Until{$|C_2-C_1|/|C_1|\leq\epsilon_2$} 
       \State 7. Output current $\bR_{opt}$, $\bQ_{opt}$ as the final limit point solution of P1.  
	\end{algorithmic}
\end{algorithm}

\section{Enhancing the Secrecy Rate Under no Eve's CSI}

In the previous section, the secrecy rate of IRS-assisted MIMO WTC is optimized based on an ideal assumption that  full CSI is available at Alice. In practice,  Eve is usually a hidden passive malicious user, it does not actively exchange its CSI with Alice, i.e., the channel matrix $\bH_{AE}$, $\bH_{IE}$ are completely unknown. Therefore, it is unlikely  to achieve secure communication by formulating an optimization
problem as ${\rm P1}$ under this case.  To the best of our knowledge, there is no current research results about how to enhance the secrecy rate in IRS-assisted system  without Eve's CSI.

 Inspired by the previous work in \cite{Wang-12}, in this section, we propose an AN aided joint transmission scheme to enahance the secrecy rate if 
$\bH_{AE}$, $\bH_{IE}$ are completely unknown at Alice. The main procedure of this scheme is concluded as two steps. In the first step, we minimize a transmit power $P_{min}$ at Alice subject to an achievable rate QoS constraint at Bob. To solve the power minimization problem, a BS and OBO combined AO algorithm is applied to   jointly optimizing $\bR$ and $\bQ$. Once the minimum power is obtained, in the second step, AN is applied   to jam Eve by using the residual power $P-P_{min}$ at Alice so as to decrease  Eve's channel capacity $C_E$. 
In addition, we also apply this AN aided scheme to a special case where the direct link between Alice and Bob/Eve is blocked (i.e., $\bH_{AB}=\bo$ and  $\bH_{AE}=\bo$).  In particular, apart from   OBO optimization method for optimizing $\bQ$ given $\bR$ for the power minimization problem, we propose an MM algorithm to obtain the sub-optimal  $\bQ$ in which all $n$ phase shift coefficients are simultaneously optimized. And we give detailed steps about how to find the proper lower bound (i.e., surrogate function) of the complicated objective function so that MM can be applied to iteratively optimize $\bQ$.

\subsection{Power Minimization and AN aided jamming}

Firstly, after obtaining the CSI of $\bH_{AI}$ and $\bH_{IB}$ at Alice, a power minimization problem subject to QoS constraint at Bob is formulated as the following ${\rm P9}$.
\bal
\notag
{\rm P9}: &\underset{\bR, \bQ}{\min}\ {\rm tr}(\bR),\\
\notag
 &s.t.\ \log_2|\bI+\bH_1\bR\bH_1^{\rm H}|\geq\gamma, |q_i|=1, \forall i, \bR\geq\bo
\eal
where $\log_2|\bI+\bH_1\bR\bH_1^{\rm H}|\geq\gamma$ is the QoS constraint, $\gamma$ is the lowest communication rate requirement at Bob. It can be known that this is also a non-convex problem  due to the non-convex QoS constraint and UMC, however, it still can be optimized via AO algorithm. Note that although \cite{Wu-19b}  also addresses to solve a power minimization problem subject to QoS constraints, its solutions only applies for the MISO case and fail to our MIMO case. Therefore, we propose an BS and OBO combined AO algorithm to address the non-convex ${\rm P9}$.

Considering $\bQ$ is fixed, the corresponding sub-problem of optimizing $\bR$ is expressed as ${\rm P10}$.
\bal
\notag
{\rm P10}: \underset{\bR}{\min}\ {\rm tr}(\bR),\  s.t.\ \log_2|\bI+\bH_1\bR\bH_1^{\rm H}|\geq\gamma, \bR\geq\bo.
\eal
To optimize $\bR$ in ${\rm P10}$, we apply the following key proposition.
\begin{prop}
Assume the optimal solution and the corresponding optimal value of {\rm P10} is $\bR_{opt}$ and $P_{opt}$ respectively,  and consider the following problem 
\bal
\notag
{\rm P10'}: &\underset{\bR}{\max}\ C'(\bR)=\log_2|\bI+\bH_1\bR\bH_1^{\rm H}|,\\
\notag
  &s.t.\  \bR\geq\bo, {\rm tr}(\bR)\leq \it{P}_{opt}.
\eal
Then, the optimal solution  and the corresponding optimal value of this problem are also $\bR_{opt}$  and $\gamma$.
\end{prop}

Proposition 2 indicates that the optimal solution  $\bR_{opt}$ for ${\rm P10}$ also solves the dual problem of maximizing the channel capacity subject to TPC ${\rm tr}(\bR)\leq P_{opt}$ in ${\rm P10'}$. This can be easily shown by contradiction or by comparing the respective KKT conditions of each problem, which are necessary for optimality \cite{Loyka-15}. Note that the optimal solution  $\bR_{opt}$ for  ${\rm P10}$ always makes the QoS constraint hold with equality. With this proposition, the work reduces to find a proper $P_{opt}$ such that the optimal value for ${\rm P10'}$ is $\gamma$. Since ${\rm P10'}$ is a general channel capacity optimization problem, the objective function $C'(\bR)$  is non-decreasing with $P_{opt}$. Hence, BS can be applied to find the proper $P_{opt}$ such that
the optimal value $C'(\bR)=\gamma$. 

With optimized $\bR$, the next step is to optimize $\bQ$ from the following sub-problem.
\bal
\notag
{\rm P11}: {\rm Find}\ \bQ,\ s.t.\ \log_2|\bI+\bH_1\bR\bH_1^{\rm H}|\geq\gamma, |q_i|=1, \forall i.
\eal
Note that there is no objective function in this problem,   any feasible $\bQ$ satisfying the  QoS and UMC can be as the optimal solution for ${\rm P11}$. In fact, if the feasible solution $\bQ$ obtained for ${\rm P11}$ achieves a strictly larger communication rate than the target rate $\gamma$, then the minimum transmit power in ${\rm P10}$ returned by BS can be properly reduced without violating the QoS constraint. Hence, the work reduces to maximize  $\log_2|\bI+\bH_1\bR\bH_1^{\rm H}|$ to be as large as possible, which is equivalent to solve the following problem ${\rm P11'}$.
\bal
{\rm P11'}: \underset{\bQ}{\max}\ \log_2|\bI+\bH_1\bR\bH_1^{\rm H}|,\  s.t.\  |q_i|=1, \forall i.
\eal 
Observing that OBO optimization Algorithm 2 can be easily applied to solve this problem by setting $\bH_{AE}=\bo, \bH_{IE}=\bo$ so that given fixed $n-1$ phase shift coefficients, the optimal solution for $q_i$ is either $e^{-j{\rm arg}(\bar{\lambda}_i)}$ or any feasible solution.

The proposed AO algorithm for solving ${\rm P9}$ is summarized as Algorithm 4. In this algorithm, since $\bR$ and $\bQ$ are  optimized independently, and also $\bR$ and $\bQ$ are both bounded by the  constraints. Hence, the objective function $\rm tr(\bR)$ is non-increasing in Algorithm 4, and a limit point solution is guaranteed to converge.

\begin{algorithm}[h]
	\caption{(\it AO algorithm of solving {\rm P9})}
	\begin{algorithmic}
		\State \bRequire\  $\epsilon_3 >0$, $\gamma$.
        \State 1. Initialize a feasible starting point $\bQ_{ini}$.
        \State 2. Solve  P10' via BS to obtain  $\bR_{ini}$ for P9 given $\bQ_{ini}$, compute $P_0={\rm tr}(\bR_{ini})$.
      \Repeat\ \ 
         \State 3. Solve P11' via OBO optimization to obtain a sub-optimal solution $\bQ_{opt}$ given $\bR_{ini}$. 
         \State 4. Solve  P10' via BS to obtain the optimal solution $\bR_{opt}$ for P9 given $\bQ_{opt}$, compute $P_1={\rm tr}(\bR_{opt})$.
		  \State 5.  If $|P_1-P_0|/|P_0|>\epsilon_3$, set $\bR_{ini}=\bR_{opt}$ and $P_0=P_1$, go back to step 3.
		\Until{$|P_1-P_0|/|P_0|\leq\epsilon_3$} 
       \State 7. Output  $\bR_{opt}$, $\bQ_{opt}$ as the final limit point solution of P9.  
	\end{algorithmic}
\end{algorithm}

After obtaining the minimum power $P_{min}={\rm tr}(\bR_{opt})$, the next step is to transmit AN to jam Eve using the residual transmit power $P-P_{min}$. To ensure the QoS at Bob is not affected by AN, we set the directions for signaling AN to $\rm{null}(\bH_1)$, i.e., the null space of the effective channel $\bW_1=\bH_1^{\rm H}\bH_1$,  and  apply equal power allocation to transmit AN to each dimension of ${\rm null}(\bW_1)$. 
Therefore, the transmit covariance for AN is formulated as
\bal
\bR_{AN}=\frac{P-P_{min}}{m-{\rm rank}(\bW_1)}\bU_{AN}\bU_{AN}^{\rm H}
\eal
where the columns in the semi-unitary matrix $\bU_{AN}$ are all $m-{\rm rank}(\bH_1)$ eigenvectors corresponding to zero eigenvalues of $\bW_1$. Hence,  the final actual achievable secrecy rate by this AN aided joint transmission scheme is 
\bal
C_s=\gamma-\log_2|\bI+\frac{\bH_2\bR\bH_2^{\rm H}}{\bI+\bH_2\bR_{AN}\bH_2^{\rm H}}|.
\eal
Note that $\bH_1$ should be full row rank matrix so that the null space of $\bW_1$ exists. Hence, our proposed scheme only holds for the case when $m>d$. If $m\leq d$, a possible  solution is that Bob can turn off some receiving antennas so as to make the number of the rest active antennas to be less than $m$, but the price is that the residual power for AN signaling could be decreased since the degree of freedom between Alice and Bob is reduced so that more transmit power needs to be consumed to meet the QoS constraint.

\subsection{A Special Case Where the Direct Link Between Alice-Bob/Eve is Blocked}

In this subsection, we consider a special case where the direct communication link between Alice-Bob and Alice-Eve  are blocked. Such case has high probability to occur in city’s hot spot, mountainous area, and other indoor environment due to  obstacles. Then, the function of IRS is to create a virtual line-of-sight link between Alice and Bob/Eve  so as to help the signals bypass the obstacle. And this is the key reason at first why IRS draws great attention by academic and industry. Obviously, the proposed AN aided joint transmission  scheme in the previous subsection can be easily applied to this case to enhance the secrecy rate by simply setting $\bH_{AB}=\bo$ and $\bH_{AE}=\bo$. In particular, apart from OBO optimization in the AO algorithm, in this subsection, we propose a MM algorithm to optimize  $\bQ$ given $\bR$ for the power minimization problem. The key  idea of MM algorithm is to firstly approximate the original non-convex problem to a more tractable formular, in which the objective function is  approximated to a linear lower  bound (i.e., surrogate function), and then the approximated problem is optimized iteratively by initializing a feasible starting point. If the bound is constructed properly, any converged point  generated by MM is a KKT point for the original problem. For detailed explanations of MM, please refer to \cite{Sun-17}. Note that different from OBO optimization, all $n$ phase shifts are simultaneously optimized in the MM algorithm.

Specifically, when $\bH_{AB}=\bo$, problem ${\rm P11'}$ transfers to the following ${\rm P12}$.
\bal
\notag
{\rm P12}: &\underset{\bQ}{\max}\ g(\bQ)=\log_2|\bI+\bH_{IB}\bQ\bH_{AI}\bR\bH_{AI}^{\rm H}\bQ^{\rm H}\bH_{IB}^{\rm H}|,\\
\notag
 &s.t.\  |q_i|=1, \forall i.
\eal
The following proposition gives the surrogate function of $g(\bQ)$ as well as the closed-form solutions of $\bQ$ during each iteration of MM algorithm.

\begin{prop}
Let $\tilde{\bQ}$ be a feasible point for ${\rm P12}$, then $g(\bQ)$ can be lower bounded by
\bal
\notag
g(\bQ)\geq &-2n\gl_1(\bZ)+2{\rm Re}\{\bq^{\rm H}(\gl_1(\bZ)\bI-\bZ)\tilde{\bq}\}+\tilde{\bq}^{\rm H}\bZ\tilde{\bq}\\
\notag
&+2{\rm Re}\{\bq^{\rm H}\ba_4\}+\sum_{j=1}^{2}C_j(\tilde{\bQ})=\tilde{g}(\bQ,\tilde{\bQ})
\eal
where $\bq=[e^{j\theta_1},e^{j\theta_2},...,e^{j\theta_n}]^{\rm T}$,  ${\rm diag}(\tilde{\bq})=\tilde{\bQ}$. Hence the closed-form solution of $\bQ$ given $\tilde{\bQ}$ during each iteration of MM algorithm is given by
\bal
\label{optimal_q}
\bQ={\rm diag}([e^{j{\rm arg}(v_1)}, e^{j{\rm arg}(v_2)},..., e^{j{\rm arg}(v_n)}]^{\rm T}).
\eal
\end{prop}
\begin{proof}
To optimize $\bQ$ in ${\rm P12}$ via MM, a proper lower bound of $g(\bQ)$ should be formulated. Note that $g(\bQ)$ is a complicated log determinant function, it is difficult to directly obtain its lower bound with only one time approximation as in the MISO case \cite{Shen-19}\cite{Yu-19}. Hence, we apply three successive approximations to obtain the proper lower bound of $g(\bQ)$ so that MM algorithm can be applied to optimize $\bQ$. Please see detailed  expression of $\bZ$, $\bq$, $C_j(\tilde{\bQ})$, $v_i, i=1,2,...,n$ as well as the proof in Appendix.
\end{proof}

Therefore, based on this proposition,  a KKT  solution of ${\rm P12}$ given fixed $\bR$ can be obtained. The BS and MM combined AO algorithm for power minimization when the direct link $\bH_{AB}=\bo$  is summarized as Algorithm 5. In this algorithm, $\bQ_{ini}$ is the starting point for the outer loop of AO algorithm and $\tilde{\bQ}$ is the starting point for the inner loop of MM algorithm. As the convergence is reached, a limit point solution for the power minimization problem can be obtained. In the MM algorithm, the main computational complexity comes from computing $\gl_1(\bZ)$ in each iteration, which is about $O(2n^3)$. Once the minimum power satisfying QoS constraint is obtained,  AN is used to signalling over the null space of $\bH_{IB}\bQ\bH_{AI}$ using the residual power at Alice so as to jam Eve. Our extensive simulation tests have shown that although Algorithm 5 have less performance on power minimization as well as enhancing secrecy rate compared with Algorithm 4, it has  faster speed of convergence with less than around 1 to 5 iterations in most randomly generated channels.

\begin{algorithm}[h]
	\caption{(\it AO algorithm for solving {\rm P12})}
	\begin{algorithmic}
		\State \bRequire\  $\epsilon_4 >0$, $\gamma$.
        \State 1. Initialize a feasible starting point $\bQ_{ini}=\bI$, $\bR_{ini}=P\bI/m$, set $P_0=P$.
      \Repeat\ \ 
          \State 2. Initialize starting point $\tilde{\bQ}=\bI$, compute $\tilde{C}_0=\tilde{g}(\tilde{\bQ},\tilde{\bQ})$ given $\bR_{ini}$.
         \Repeat\ \ (MM algorithm)
                \State 3. Optimize $\bQ$ via \eqref{optimal_q}, and compute $\tilde{C}_1=\tilde{g}(\bQ,\tilde{\bQ})$.
                 \State 4. If $|\tilde{C}_1-\tilde{C}_0|/|\tilde{C}_0|$ does not converge, set $\tilde{C}_0=\tilde{C}_1$ and $\tilde{\bQ}=\bQ$.    
              \Until{$|\tilde{C}_1-\tilde{C}_0|/|\tilde{C}_0|$ converges} 
              \State 5. Set $\bQ_{opt}=\bQ$ as the KKT solution of P12.
         \State 6.  Obtain the optimal solution $\bR_{ini}$ via BS given $\bQ_{opt}$, compute $P_1={\rm tr}(\bR_{opt})$.
		  \State 7.  If $|P_1-P_0|/|P_0|>\epsilon_4$, set $\bR_{ini}=\bR_{opt}$ and $P_0=P_1$, go back to step 3.
		\Until{$|P_1-P_0|/|P_0|\leq\epsilon_4$}  
	\end{algorithmic}
\end{algorithm}

\section{Simulation Results}

To validate the performance of our proposed AO algorithms,  extensive simulation results have been carried out in this section. Following \cite{Shen-19},  we consider a fading environment, and all the channels are  formulated as   the product of large scale fading and small scale fading. The entries in the small scale fading matrix are   randomly generated  with complex zero-mean
Gaussian random variables with unit covariance. For the large scale fading in all links, the path loss is set as -30dB at reference distance 1m, and path loss exponents for all the links is set as 3. And we assume  that the distance between Alice and Bob, Alice and IRS, Alice and Eve, IRS and Bob, IRS and Eve are set as 80m, 30m, 80m, 40m and 40m respectively. In AO Algorithm 3, 4 and 5, we set all the target accuracy  as $\epsilon_2=\epsilon_3=\epsilon_4=10^{-4}$, and all the target accuracy  for BS algorithm as $10^{-4}$. In addition,  $\epsilon_1=10^{-8}$, $\alpha=0.3$,  $\beta=0.5$, $\eta=5$, $t_0=10^2$, $t_{max}=10^5$ in Algorithm 1 and target accuracy for MM is $10^{-4}$ in Algorithm 5. Note that all the simulation results illustrated in Fig.2 to Fig.4 and Fig.6 are averaged over 100 randomly generated channels, and all the results in Fig.5, Fig.7 to Fig.10 are computed based on single randomly generated channel. 

\subsection{Secrecy Rate of IRS-Assisted MIMO WTC Under Full CSI}

In this subsection, the performance of proposed AO Algorithm 3 for maximizing the secrecy rate under full CSI is provided. In Fig.2, we compare the average secrecy rate performance of our AO Algorithm 3
with three benchmark schemes: 1). optimize $\bR$ via  Algorithm 1 given zero phase
shift (i.e., $\bQ = \bI$) at IRS; 2). optimize $\bR$  via  Algorithm 1 without  IRS (i.e., $\bQ = \bo$); 3). AN aided solutions\footnote{We remark that although the system model illustrated in \cite{Fang-15} is a cognitive radio MIMO WTC with simultaneous wireless
information and power transfer, its proposed AN method with sub-optimal algorithm also applies for the  MIMO
WTC model.} without IRS in \cite{Fang-15}. According to the result, we note that our Algorithm 3 has significantly better performance than the other three benchmark schemes. The main reason is that the reflecting link signals are not only constructively added with the direct link signals at Bob, but also destructively added with the direct link signals at Eve and hence the secrecy rate can be boosted. For the  three benchmark schemes, it can be seen that  the solution with and without AN under no IRS have very limited performance on enhancing secrecy rate. Furthermore,  although zero phase  shift solutions with IRS have better performance than that without IRS, it still has a large gap compared with the results returned by Algorithm 3. The reason is  that the phase shift $\bQ$ is unchanged, i.e.,  IRS doesn't truly change the propagation channels. Therefore,  only by jointly optimizing  $\bR$ and $\bQ$ can we give full play to the advantages of IRS on enhancing the secrecy performance.

\begin{figure}[t]
	\centerline{\includegraphics[width=3.0in]{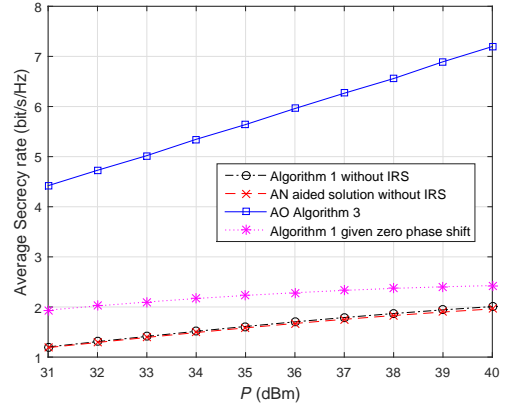}}
	\caption{  Secrecy rate via proposed  AO Algorithm 3 and other benchmark schemes based on $m=d=e=4, n=6$.}
\end{figure}

Fig.3 illustrate the performance of Algorithm 3 and other benchmark schemes versus the number of reflecting element $n$  at IRS and number of antenna $e$ at Eve respectively. Note that  the secrecy rate returned by Algorithm 3 increases with $n$ significantly since more reflecting elements brings more new spacial degree of freedom. Furthermore, for the zero phase shift scheme in which only $\bR$ is optimized, only by increasing $n$ has very limited performance gain on the secrecy rate, even the secrecy rate decreases with $n$. The main reason is that unoptimized $\bQ$ is possible to make the Eve's effective channel $\bH_2$ ``better" than Bob's effective channel $\bH_1$.  We also note that  the secrecy rate returned by Algorithm 3 and other benchmark schemes all decrease with $e$. This is inevitable since with more value of $e$,  the   sufficient spatial degree of freedom for Alice's transmission is decreased. However, as can be seen, given fixed $e$, larger secrecy rate still can be obtained by our proposed AO algorithm for  the IRS-assisted design than the other solutions. In fact, if Eve is equipped with more antennas, an effective solution is to deploy more reflecting elements at IRS so as to enhance the secrecy rate. This can be easily realized in practical system, since the reflecting elements in IRS are with very low complexity, low power consumption and can be massively deployed.

\begin{figure}[t]
	\centerline{\includegraphics[width=3.0in]{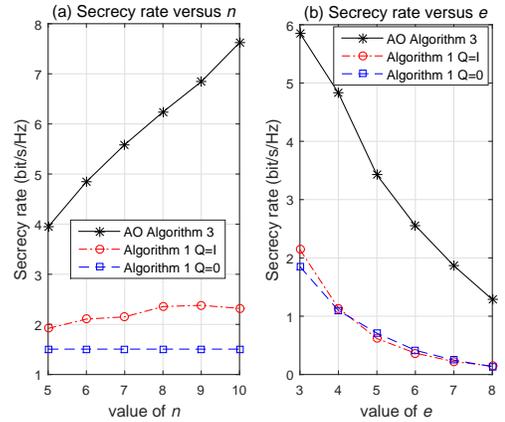}}
	\caption{Achieved secrecy rate versus the number of $n$ and $e$, $m=4, d=e=3$ in (a), and $m=d=3, n=6$ in (b).}
\end{figure}

\subsection{Secrecy Rate of IRS-Assisted MIMO WTC Under No Eve's CSI}

In this subsection, we provide some simulation results about the performance of proposed AN aided joint transmission scheme on enhancing the secrecy rate under completely no Eve's CSI. Fig.4 shows the performance of averaged secrecy rate returned by the scheme with and without AN versus the target QoS $\gamma$ at Bob under different settings of $m, n$. Based on the averaged results, note that   positive secrecy rate (i.e., $C_B>C_E$) also can be achieved by our proposed scheme so that secure communication can be guaranteed. Also note that without the aid of AN, only guaranteeing QoS  at Bob has very limited performance on enhancing the secrecy rate. In this scheme, as long as $\gamma$ is properly set,  Alice could have abundant residual power to jam Eve via AN signalling so that the channel capacity at Alice can be larger than that at Eve. In fact, based on our extensive simulation results, when $\gamma=6$, only about less than $30\%$ of the total transmit power is consumed to meet the QoS constraint by Algorithm 4, and more than half of the total power are all utilized to jam Eve via  AN, resulting positive secrecy rate.

\begin{figure}[t]
	\centerline{\includegraphics[width=3.0in]{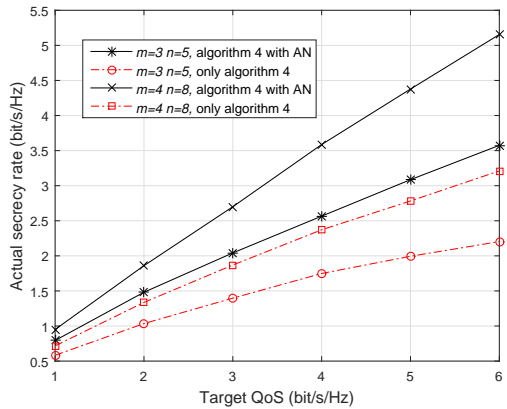}}
	\caption{Achieved actual secrecy rate versus Target $\gamma$ under different settings of $m$ and $n$. Total transmit power at Alice is set as $35 {\rm dBm}$.}
\end{figure}

To exploit how the secrecy performance returned by the AN aided scheme goes when the
target $\gamma$ is large under no Eve’s CSI, Fig.5 shows the actual secrecy rate versus $\gamma$ under
different settings of total power $P$. Note that for each setting of $P$, the secrecy rate firstly
increases with $\gamma$, since the transmission rate $C_B =\gamma$ at Bob dominates so that $C_s$ increases with
$\gamma$. However, as $\gamma$ grows to higher, $C_s$ starts to decrease, since the sufficient residual
power $P-P_{min}$ for AN signaling is reduced significantly so that the information leakage to Eve
$C_E$ dominates. Finally, as $\gamma$ grows to high enough, the total power $P$ can not support to meet the QoS constraint so that ${\rm P9}$ becomes infeasible and hence secure communication is not achievable (e.g., the
secrecy rate stops at $\gamma=10$ when $P=30 {\rm dBm}$). Therefore, we see that there is a trade off
between increasing QoS at Bob and enhancing the secrecy performance. In addition to this result
based on single channel realization, such trade off also exists in our other extensive simulations.
Therefore, it is better to balance the setting between $\gamma$ and the
residual power $P-P_{min}$ so as to achieve a good secrecy performance.

\begin{figure}[t]
	\centerline{\includegraphics[width=3.0in]{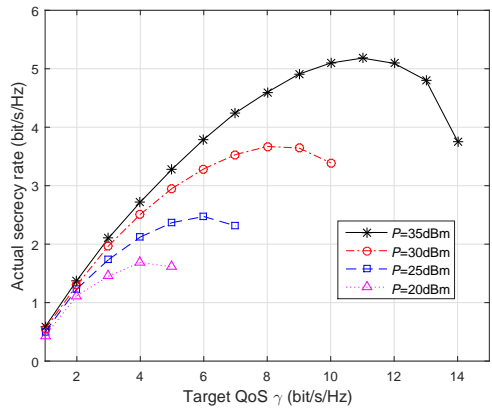}}
	\caption{ The actual secrecy rate returned by the proposed scheme for no Eve’s CSI versus $\gamma$ given different total power $P$ at Alice . The channels are randomly generated via $m=e=4, d=2, n=8$.}
\end{figure}

In Fig.6, we assume a special case where  the direct link $\bH_{AB}=\bo, \bH_{AE}=\bo$, and show the performance of averaged secrecy rate as well as minimized power (to meet the QoS constraint) returned by Algorithm 4 and Algorithm 5 versus $\gamma$. Observing that both Algorithms could guarantee positive secrecy rate, and Algorithm 4 achieves better performance than Algorithm 5, the main reason is that the optimized minimum power returned by Algorithm 4 is less than by Algorithm 5 (see ``(b)Minimum transmit power" in Fig.6) and hence more residual power for AN signaling can be saved via Algorithm 4. In addition to this results, our extensive simulation tests based on different channel realizations show that a  larger  actual secrecy rate   can be obtained by Algorithm 4 than that by Algorithm 5 at finite threshold of $\gamma$.

\begin{figure}[t]
	\centerline{\includegraphics[width=3.0in]{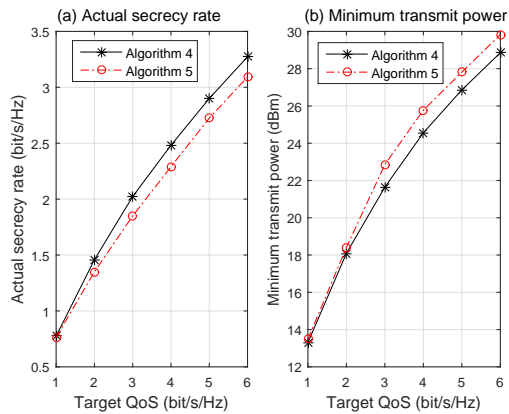}}
	\caption{ Actual secrecy rate and minimized power versus Target $\gamma$ under $\bH_{AB}=\bH_{AE}=\bo$, $m=4, n=8, d=e=2$. Total transmit power at Alice is set as $35 {\rm dBm}$.}
\end{figure}

\subsection{Convergence of the Proposed AO Algorithms}

In this subsection, we provide some numerical examples about the convergence of proposed AO Algorithm 3, 4 and 5 under different value of $m,n,d,e$ and target threshold $\gamma$. Fig.7 illustrates the convergence of the objective function
$C_s(\bR_k,\bQ_k)$ versus the number of iterations $k$ in Algorithm 3 under randomly generated channels. Based on the results, it requires 19 to 101 steps for $C_s(\bR_k,\bQ_k)$ to
converge to the accuracy of $10^{-4}$ for each considered setting, also note that
the process of convergence is monotonically increasing. In fact, given
fixed target accuracy, larger settings of $m$ and $n$ leads to larger
dimensions of variable $\bR$ and $\bQ$ so that the  algorithm
requires more iterations to optimize each element of these
variables. In addition to these results, our extensive simulations
show a perfect monotonic convergence of Algorithm 3.

\begin{figure}[t]
	\centerline{\includegraphics[width=3.0in]{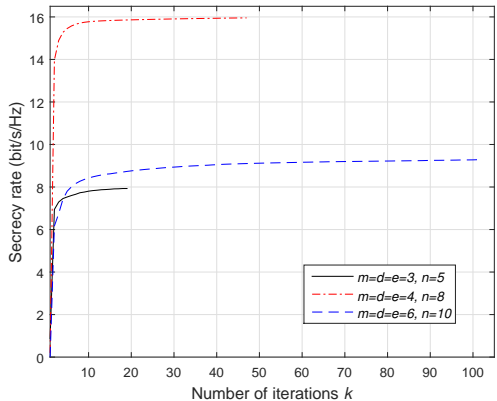}}
	\caption{Convergence of the proposed Algorithm 3 under different settings of $m, n, d, e$. $\gamma$ is fixed as 3.}
\end{figure}

Fig.8 illustrate the convergence of  Algorithm 4  for solving ${\rm P9}$. Note that  the objective function ${\rm tr}(\bR_k)$ is non-increasing with number of iteration $k$ and finally converge. Furthermore, the convergence is fast with only 5 to 8 steps to reach the target accuracy. Fig.9 compares the convergence between Algorithm 4 and 5 for solving ${\rm P9}$ under $\bH_{AB}=\bo$.  As can be seen, Algorithm 5 has a slightly faster convergence with less than 2 to 3 steps compared with Algorithm 4. But when both algorithm converges, the minimum power returned by Algorithm 5 is still less than that by  Algorithm 4, which  is consistent with those results in Fig.6. In addition to this results, our other extensive simulation tests indicates that a faster speed of convergence (with less than around 1 to 5 steps) than Algorithm 4 can be achieved.

\begin{figure}[t]
	\centerline{\includegraphics[width=3.0in]{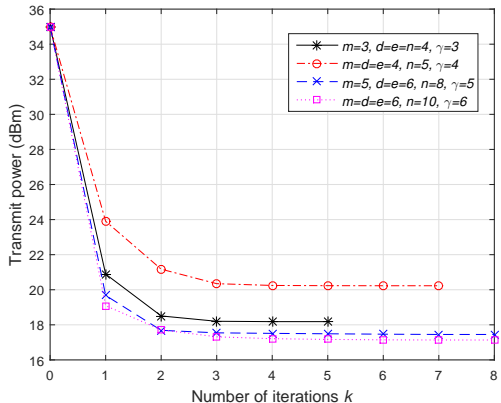}}
	\caption{ Convergence of the proposed Algorithm 4 under different settings $m, n, d, e$ and $\gamma$. Total transmit power at Alice is set as $35 {\rm dBm}$.}
\end{figure}

\begin{figure}[t]
	\centerline{\includegraphics[width=3.0in]{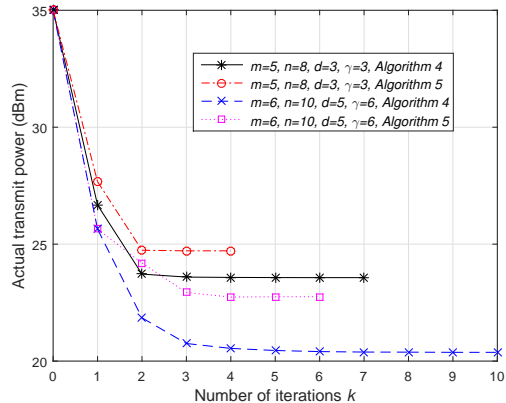}}
	\caption{Convergence comparison between  Algorithm 4 and 5 when $\bH_{AB}=\bH_{AE}=\bo$ under different settings of $m, n, d$ and $\gamma$. Total transmit power at Alice is set as $35 {\rm dBm}$.}
\end{figure}

 Finally, Fig.10 shows the convergence of  the value of objective function $f(\bR,\bK)$ and $C(\bR)$ in Algorithm 2 versus  number of iterations $k$. It can be noted that both $f(\bR,\bK)$ and $C(\bR)$ gradually converge and coincide together as $t$ increases to $t_{max}$. In particular, it
takes few more steps for $C(\bR)$ to converge than $f(\bR,\bK)$, which indicates that $C(\bR)$ is less
sensitive in $\bR$ than that in $f(\bR,\bK)$. Hence, it is necessary to set $t_{max}$ large enough in barrier method so that the optimized $\bR$ for P3 is guaranteed to be a global optimal solution for P2. Furthermore, this results also have validated the correctness of our re-derived gradient and Hessians of the barrier objective function $f_t(\bR,\bK)$ under complex-valued channel case shown in Appendix.

\begin{figure}[t]
	\centerline{\includegraphics[width=3.0in]{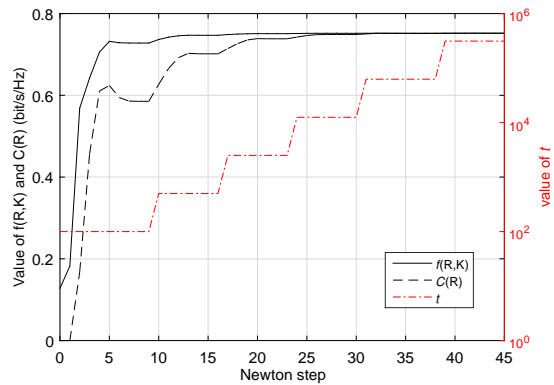}}
	\caption{ Convergence of $f(\bR,\bK)$ and $C(\bR)$ in Algorithm 1 when $m=d=e=4, n=6$ and $P=35{\rm dBm}$. As $t$ increases to the threshold $t_{max}$, both $f(\bR,\bK)$ and $C(\bR)$ converges and coincide together.}
\end{figure}

\section{Concluding Remarks and Future Directions}
In this paper, the secrecy rate optimization of IRS-assisted Gaussian MIMO WTC is studied.  When full CSI is assumed, a barrier method and OBO optimization combined AO algorithm is proposed to   jointly maximize the transmit covariance $\bR$ at Alice as well as phase shift coefficient $\bQ$ at IRS. When no Eve's CSI is assumed, we propose an AN aided joint transmission scheme to enhance the secrecy rate.  In this scheme, a BS and OBO combined AO algorithm is proposed to jointly optimize $\bR$ and $\bQ$. This AN aided scheme is further applied to enhance the secrecy rate under a special scenario where the direct communication between Alice and Bob/Eve is blocked. In particular, a BS and MM combined AO algorithm with slightly faster convergence is proposed to jointly optimize $\bR$ and $\bQ$. Simulation results  have validated the monotonic convergence of the proposed algorithms, and it is shown that the proposed algorithms for the IRS-assisted design achieve significantly larger secrecy rate than the other benchmark schemes with and without IRS under full CSI. When Eve's CSI is unknown, the secrecy performance can also be greatly enhanced by the proposed AN aided joint transmission scheme.

A future direction of our work is to exploit the solutions for enhancing the robust secrecy rate of IRS-assisted MIMO WTC when the CSI of the direct and reflecting channel links are imperfectly known at Alice due to estimation errors. This is a more challenging study  since in the secrecy rate optimization problem, the objective function as well as the constraints become more complicated due to the  bounded channel estimation errors.  Therefore, new numerical algorithm should be redesigned to jointly optimize $\bR$ and $\bQ$ as well as its proof of convergence. The study of robust secure IRS-assisted MIMO channel also provides significant references for the application in practical  system design.

\section{Appendix}

\subsection{Closed-form Expressions for Gradients and Hessians in Algorithm 1}

After some manipulations, we provide below the analytical expressions for the gradients  and
Hessians obtained,  using the standard rules of matrix differential calculus (see Chapter 3 to 6 in \cite{Hjorungnes-11}). The gradient of  $f_t(\bR,\bK)$ respect to $\br$ and $\bn$ are expressed as
\bal
\notag
&\nabla_{\br^{*}} f_t(\bR,\bK)=\bD_{\br}^{\rm T}{\rm vec}((\bI+\bH^{\rm H}\bK^{-1}\bH\bR)^{-1}\bH^{\rm H}\bK^{-1}\bH-(\bI
\\
\notag
&+\bH_2^{\rm H}\bH_2\bR)^{-1}\bH_2^{\rm H}\bH_2+t^{-1}\bR^{-1}-t^{-1}(P-{\rm tr}(\bR))^{-1}\bI),\\
\notag
&\nabla_{\bn^{*}} f_t(\bR,\bK)=\bD_{\bn}^{\rm T}{\rm vec}((\bK+\bH\bR\bH^{\rm H})^{-1}-(1+t^{-1})\bK^{-1}).
\eal
 Then, each entry of $\bT$ is expressed as
\bal
\notag
&\nabla_{\br^{*}\br^{\rm T}}^2 f_t(\bR,\bK)=-\bD_{\br}^{\rm T}(\bZ_1^{\rm T}\otimes\bZ_1-\bZ_2^{\rm T}\otimes\bZ_2\\
\notag
&+t^{-1}\bR^{-{\rm T}}\otimes\bR^{-1}+t^{-1}(P-{\rm tr}(\bR))^{-2}{\rm vec}(\bI){\rm vec}(\bI)^{\rm T})\bD_{\br},\\
\notag 
&\nabla_{\bn^{*}\bn^{\rm T}}^2 f_t(\bR,\bK)=-\bD_{\bn}^{\rm T}(\bZ_3^{\rm T}\otimes\bZ_3\\
\notag 
&\quad\quad\quad\quad\quad\quad\quad\quad-(1+t^{-1})\bK^{-{\rm T}}\otimes\bK^{-1})\bD_{\bn},\\
\notag
&\nabla_{\br^{*}\bn^{\rm T}}^2 f_t(\bR,\bK)=-\bD_{\br}^{\rm T}(\bZ_4^{\rm T}\otimes\bZ_4^{\rm H})\bD_{\bn}
\eal
where $\bZ_1=(\bI+\bH^{\rm H}\bK^{-1}\bH\bR)^{-1}\bH^{\rm H}\bK^{-1}\bH$, $\bZ_2=(\bI+\bH_2^{\rm H}\bH_2\bR)^{-1}\bH_2^{\rm H}\bH_2$, $\bZ_3=(\bK+\bH\bR\bH^{\rm H})^{-1}$, $\bZ_4=\bH^{\rm H}(\bK+\bH\bR\bH^{\rm H})^{-1}$.

\subsection{Proof of Proposition 1}
Firstly, by applying the key steps of deriving (9) and (15) in \cite{Zhang-19},  $C_B$ and $C_E$ given $\bR$ can be  expressed as
\bal
\notag
&C_B=\log_2|\bI+q_i\bA_i^{-1}\bB_i+q_i^*\bA_i^{-1}\bB_i^{\rm H}|+\log_2|\bA_i|,\\
\notag
&C_E=\log_2|\bI+q_i\bC_i^{-1}\bD_i+q_i^*\bC_i^{-1}\bD_i^{\rm H}|+\log_2|\bC_i|
\eal
where
\bal 
\notag
\bA_i=&\bI+(\bar{\bH}_{AB}+\sum_{j=1,j\neq i}^{n}q_j\br_j\bar{\bt}_j^{\rm H})(\bar{\bH}_{AB}+\sum_{j=1,j\neq i}^{n}q_j\br_j\bar{\bt}_j^{\rm H})^{\rm H}\\
\notag
&+\br_i\bar{\bt}_i^{\rm H}\bar{\bt}_i\br_i^{\rm H}, \\
\notag
\bB_i=&\br_i\bar{\bt}_i^{\rm H}(\bar{\bH}_{AB}^{\rm H}+\sum_{j=1,j\neq i}^{n}\bar{\bt}_j\br_j^{\rm H}q_j^*),\\
\notag
\bC_i=&\bI+(\bar{\bH}_{AE}+\sum_{j=1,j\neq i}^{n}q_j\bs_j\bar{\bt}_j^{\rm H})(\bar{\bH}_{AE}+\sum_{j=1,j\neq i}^{n}q_j\bs_j\bar{\bt}_j^{\rm H})^{\rm H}\\
\notag
&+\bs_i\bar{\bt}_i^{\rm H}\bar{\bt}_i\bs_i^{\rm H},\\
\notag
\bD_i=&\bs_i\bar{\bt}_i^{\rm H}(\bar{\bH}_{AE}^{\rm H}+\sum_{j=1,j\neq i}^{n}\bar{\bt}_j\bs_j^{\rm H}q_j^*)
\eal 
and where $\bar{\bH}_{AB}=\bH_{AB}\bU_{\bR}\bSigma_{\bR}^{0.5}$,  $\bar{\bH}_{AE}=\bH_{AE}\bU_{\bR}\bSigma_{\bR}^{0.5}$, $\bar{\bt}_i$ is the $i$-th column of $\bar{\bH}_{AI}^{\rm H}=[\bH_{AI}\bU_{\bR}\bSigma_{\bR}^{0.5}]^{\rm H}$, $\br_i$ and $\bs_i$ are the $i$-th column of $\bH_{IB}$ and $\bH_{IE}$ respectively.  
Since $\bA_i$ and $\bC_i$ are both full rank matrices and ${\rm rank}(\bB_i)\leq 1$, ${\rm rank}(\bD_i)\leq 1$, ${\rm rank}(\bA_i^{-1}\bB_i)\leq 1$ and ${\rm rank}(\bC_i^{-1}\bD_i)\leq 1$ hold. 
Hence, by dropping the constant term $\log_2|\bA_i|$ and $\log_2|\bC_i|$,  ${\rm P5}$ reduces to ${\rm P6}$. It can be known that $\bA_i^{-1}\bB_i$ and $\bC_i^{-1}\bD_i$ are strongly related to $q_i$. 
In the following, we consider 4 cases about the trace of $\bA_i^{-1}\bB_i$ and $\bC_i^{-1}\bD_i$ and propose the optimal solutions of $q_i$ based on each case.

   1). Case 1: ${\rm tr}(\bA_i^{-1}\bB_i)={\rm tr}(\bC_i^{-1}\bD_i)=0$. In this case, both $\bA_i^{-1}\bB_i$ and $\bC_i^{-1}\bD_i$ are non diagonalizable according to Lemma 2 in \cite{Zhang-19}. Therefore, by applying the  steps of deriving (21) in  \cite{Zhang-19}, the objective function in ${\rm P6}$ is further derived as
\bal
\notag
C'_B(q_i)-C'_E(q_i)=&\log_2|\bI-\bA_i^{-1}(\bA_i^{-1}\bB_i)^{\rm H}\bB_i|\\
\notag
&-\log_2|\bI-\bC_i^{-1}(\bC_i^{-1}\bD_i)^{\rm H}\bD_i|.
\eal
Obviously, this objective function  is irrelevant to $q_i$. Hence, any solution of $q_i$  satisfying $|q_i|=1$ is the optimal solution of ${\rm P6}$.

 2). Case 2: ${\rm tr}(\bA_i^{-1}\bB_i)\neq0, {\rm tr}(\bC_i^{-1}\bD_i)=0$. In this case, $\bA_i^{-1}\bB_i$ is diagonalizable and  $C'_E(q_i)$ is not affected by $q_i$. Let the eigenvalue decomposition of $\bA_i^{-1}\bB_i$ as $\bA_i^{-1}\bB_i=\bar{\bU}_i\bar{\bSigma}_i\bar{\bU}_i^{\rm H}$, where $\bar{\bU}_i$ is the unitary matrix, the columns of which are the eigenvectors, $\bar{\bSigma}_i={\rm diag}\{\bar{\lambda}_i,0,...0\}$, and let $\bar{\bV}_i=\bar{\bU}_i^{\rm H}\bA_i\bar{\bU}_i$, $\bar{\bv}_i$ denotes the first column of $\bar{\bV}_i^{-1}$, $\bar{\bv}_i^{'{\rm T}}$ denotes the first row of $\bar{\bV}_i$.  By applying the key steps of deriving (16) and (17) in \cite{Zhang-19}, $C'_B(q_i)$ is further expressed as
\bal
\label{C'_B}
C'_B(q_i)=\log_2(1+\bar{\lambda}_i^2(1-\bar{v}_{i1}^{'}\bar{v}_{i1})+2{\rm Re}\{q_i\bar{\lambda}_i\})
\eal
where $\bar{v}_{i1}$ and $\bar{v}_{i1}^{'}$ are the first elements in $\bar{\bv}_i$ and $\bar{\bv}_i^{'{\rm T}}$ respectively. Therefore, it can be directly obtained that the optimal solution for maximizing $C'_B(q_i)$ is  $q_i=e^{-jarg(\bar{\lambda}_i)}$.

 3). Case 3: ${\rm tr}(\bA_i^{-1}\bB_i)=0, {\rm tr}(\bC_i^{-1}\bD_i)\neq0$. In this case, $\bC_i^{-1}\bD_i$ is diagonalizable and  $C'_B(q_i)$ is not affected by $q_i$. Let the eigenvalue decomposition of $\bC_i^{-1}\bD_i$ as $\bC_i^{-1}\bD_i=\tilde{\bU}_i\tilde{\bSigma}_i\tilde{\bU}_i^{\rm H}$, where $\tilde{\bU}_i$ is the unitary matrix, the columns of which are the eigenvectors, $\tilde{\bSigma}_i={\rm diag}\{\tilde{\lambda}_i,0,...0\}$, and let $\tilde{\bV}_i=\tilde{\bU}_i^{\rm H}\bC_i\tilde{\bU}_i$, $\tilde{\bv}_i$ denotes the first column of $\tilde{\bV}_i^{-1}$, $\tilde{\bv}_i^{'{\rm T}}$ denotes the first row of $\tilde{\bV}_i$. Similar with \eqref{C'_B}, $C'_E(q_i)$ can be further expressed as
\bal
\label{C'_E}
C'_E(q_i)=\log_2(1+\tilde{\lambda}_i^2(1-\tilde{v}_{i1}^{'}\tilde{v}_{i1})+2{\rm Re}\{q_i\tilde{\lambda}_i\})
\eal
where $\tilde{v}_{i1}$ and $\tilde{v}_{i1}^{'}$ are the first elements in $\tilde{\bv}_i$ and $\tilde{\bv}_i^{'{\rm T}}$ respectively. Hence, to maximize  P6, the work reduces to minimize  $C'_E(q_i)$, i.e., ${\rm Re}\{q_i\bar{\lambda}_i\}$ should be minimized. Therefore, it can be directly obtained that the optimal solution for minimizing $C'_E(q_i)$  is $q_i=e^{j(\pi-arg(\tilde{\lambda}_i))}$. 

4). Case 4: ${\rm tr}(\bA_i^{-1}\bB_i)\neq0, {\rm tr}(\bC_i^{-1}\bD_i)\neq0$. In this case, both  $\bA_i^{-1}\bB_i$ and  $\bC_i^{-1}\bD_i$ are diagonalizable so that the  both $C'_B(q_i)$ and $C'_E(q_i)$ are related to $q_i$. By combining \eqref{C'_B} and \eqref{C'_E} together, ${\rm P6}$ can be further expressed as 
\bal
{\rm P7}: \underset{q_i}{\max}\  \frac{\bar{c}_i+2{\rm Re}\{q_i\bar{\lambda}_i\}}{\tilde{c}_i+2{\rm Re}\{q_i\tilde{\lambda}_i\}},\ s.t.\  |q_i|=1, \forall i
\eal
where $\bar{c}=1+\bar{\lambda}_i^2(1-\bar{v}_{i1}^{'}\bar{v}_{i1}), \tilde{c}=1+\tilde{\lambda}_i^2(1-\tilde{v}_{i1}^{'}\tilde{v}_{i1})$.
P7 is a fractional programming optimization problem, which can be solved via Dinkelback method. Let $u\geq0$, and then we consider the following problem.
\bal
\notag
{\rm P8}: &\underset{q_i}{\max}\ f(q_i/u)=\bar{c}_i+2{\rm Re}\{q_i\bar{\lambda}_i\}-u(\tilde{c}_i+2{\rm Re}\{q_i\tilde{\lambda}_i\}),\\
\notag
&s.t. \ |q_i|=1, \forall i.
\eal
 Hence, according to the key concept of Dinkelbach method, finding the optimal solution of ${\rm P7}$ is equivalently  to finding the proper value of $u$ such that the optimal value of the objective function  $f(q_i/u)$ in P8 is zero. To obtain the proper $u$, the following key lemma is needed.
\begin{lemma}
Consider $q_{i}^{opt}$ is the optimal solution given fixed $u$ in P8, then $f(q_{i}^{opt}/u)$ is a monotonically decreasing function in $u$.
\end{lemma}
\begin{proof}
Assume $0<u_1<u_2$, and denote $q_{i}^{opt1}$ and $q_{i}^{opt2}$ to be the optimal solution of P8 given $u_1$ and $u_2$ respectively, then
\bal
\notag
f(q_{i}^{opt2}/u_2)&=\bar{c}_i+2{\rm Re}\{q_{i}^{opt2}\bar{\lambda}_i\}-u_2(\tilde{c}_i+2{\rm Re}\{q_{i}^{opt2}\tilde{\lambda}_i\})\\
\notag
&<\bar{c}_i+2{\rm Re}\{q_{i}^{opt2}\bar{\lambda}_i\}-u_1(\tilde{c}_i+2{\rm Re}\{q_{i}^{opt2}\tilde{\lambda}_i\})\\
\notag
&\leq \bar{c}_i+2{\rm Re}\{q_{i}^{opt1}\bar{\lambda}_i\}-u_1(\tilde{c}_i+2{\rm Re}\{q_{i}^{opt1}\tilde{\lambda}_i\})\\
&=f(q_{i}^{opt1}/u_1).
\eal
Hence, $f(q_{i}^{opt2}/u_2)<f(q_{i}^{opt1}/u_1)$, from which the proof is complete.
\end{proof}
Based on this key lemma, BS algorithm can be applied to find the optimal $u$ such that $f(q_{i}^{opt}/u)=0$. Note that $f(q_i/u)$ can be further expressed as
$f(q_i/u)=\bar{c}_i-u\tilde{c}_i+2{\rm Re}\{q_i(\bar{\lambda}_i-u\tilde{\lambda}_i)\}$.
Hence, once the optimal $u$ is obtained, it can be directly obtained that the optimal solution for minimizing $C'_E(q_i)$  is $q_i=e^{-jarg(\bar{\lambda}_i-u\tilde{\lambda}_i)}$, from which the proof is complete.

\subsection{Proof of Proposition 3}

Firstly,   let $\bP=\bH_{IB}\bQ\bL^{\frac{1}{2}}$, $\bL=\bH_{AI}\bR\bH_{AI}^{\rm H}$, according to matrix inversion lemma, $g(\bQ)$ can be further expressed as $g(\bQ)=-\log_2|\bI-\bP(\bI+\bP^{\rm H}\bP)^{-1}\bP^{\rm H}|$.
To find the lower bound of $g(\bQ)$, we firstly introduce the following key lemma.
\begin{lemma}
For any matrix $\bA\in\mathbb{C}^{m\times m}$ and $\tilde{\bA}\in\mathbb{C}^{m\times m}$, 
\bal
\label{lemma_logA}
\log_2|\bA| \leq \log_2|\tilde{\bA}|+{\rm tr}(\tilde{\bA}^{-1}(\bA-\tilde{\bA})).
\eal
\end{lemma}
\eqref{lemma_logA} holds since  $\log_2|\bA|$ is concave in $\bA$.  Hence, let $\bQ_B=\bI-\bP(\bI+\bP^{\rm H}\bP)^{-1}\bP^{\rm H}$ and consider $\tilde{\bQ}$ is a feasible point satisfying the unit modulus constraint, then $g(\bQ)$ can be lower bounded by
\bal
\notag
g(\bQ)&\geq -\log_2|\tilde{\bQ}_B|-{\rm tr}(\tilde{\bQ}_B^{-1}(\bQ_B-\tilde{\bQ}_B))\\
\notag
&=C_1(\tilde{\bQ})+h_B(\bQ)
\eal
where  $C_1(\tilde{\bQ})=-\log_2|\tilde{\bQ}_B|+{\rm tr}(\bI)-{\rm tr}(\tilde{\bQ}_B^{-1})$, $h_B(\bQ)={\rm tr}(\tilde{\bQ}_B^{-1}\bP(\bI+\bP^{\rm H}\bP)^{-1}\bP^{\rm H})$, $\tilde{\bQ}_B=\bI-\tilde{\bP}(\bI+\tilde{\bP}^{\rm H}\tilde{\bP})^{-1}\tilde{\bP}^{\rm H}$, $\tilde{\bP}=\bH_{IB}\tilde{\bQ}\bL^{\frac{1}{2}}$.
It can be verified that $C_1(\tilde{\bQ})+h_B(\bQ)$ is a surrogate function of $g(\bQ)$ so that P12 is approximated to the following ${\rm P13}$.
\bal
{\rm P13}: \underset{\bQ}{\max}\ C_1(\tilde{\bQ})+h_B(\bQ),\ s.t., |q_i|=1.
\eal
However, it is still difficult to apply  MM algorithm to solve this problem due to the complicate structure of $h_B(\bQ)$ as well as non-convex UMC. Hence, we apply a second approximation of $g(\bQ)$ by finding a lower bound of $h_B(\bQ)$. The following key lemma of matrix fractional functions is need to construct this bound \cite{Boyd-04}.

\begin{lemma}
For any positive semi-definite matrix $\bA\in\mathbb{C}^{m \times m}$ and positive definite matrix $\bB, \tilde{\bB}\in\mathbb{C}^{n \times n}$, and $\bX, \tilde{\bX}\in\mathbb{C}^{m \times n}$, 
\bal
\notag
&{\rm tr}(\bA\bX\bB^{-1}\bX^{\rm H})\\
\notag
\geq &{\rm tr}(\bA\tilde{\bX}\tilde{\bB}^{-1}\tilde{\bX}^{\rm H})-{\rm tr}(\bA\tilde{\bX}\tilde{\bB}^{-1}(\bB-\tilde{\bB})\tilde{\bB}^{-1}\tilde{\bX}^{\rm H})\\
\notag
&+{\rm tr}(\bA(\bX-\tilde{\bX})\tilde{\bB}^{-1}\tilde{\bX}^{\rm H})+{\rm tr}(\bA\tilde{\bX}\tilde{\bB}^{-1}(\bX-\tilde{\bX})^{\rm H}).
\eal
\end{lemma}
Lemma 3 holds since ${\rm tr}(\bA\bX\bB^{-1}\bX^{\rm H})$ is equivalent to the sum of m matrix fractional functions, and thus convex \cite{Boyd-04}. Therefore, by applying this lemma to the term $h_B(\bQ)$  via setting $\bA=\tilde{\bQ}_B^{-1}$, $\bX=\bP$, $\tilde{\bX}=\tilde{\bP}=\bH_{IB}\tilde{\bQ}\bL^{\frac{1}{2}}$, $\bB=\bI+\bP^{\rm H}\bP$ and $\tilde{\bB}=\bI+\tilde{\bP}^{\rm H}\tilde{\bP}$ and after some manipulations, the lower bound of $g(\bQ)$ can be further expressed as
\bal
\label{second}
g(\bQ)\geq C_1(\tilde{\bQ})+h_B(\bQ)\geq \sum_{i=1}^{2}C_i(\tilde{\bQ})+g_B(\bQ)
\eal
where 
$C_2(\tilde{\bQ})=-{\rm tr}(\tilde{\bQ}_B^{-1})+{\rm tr}(\tilde{\bQ}_B^{-1}\bJ_B\tilde{\bP}^{\rm H}\tilde{\bP}\bJ_B^{\rm H})$, $g_B(\bQ)=-{\rm tr}(\tilde{\bQ}_B^{-1}\bJ_B\bP^{\rm H}\bP\bJ_B^{\rm H})+{\rm tr}(\tilde{\bQ}_B^{-1}\bJ_B\bP^{\rm H})+{\rm tr}(\bP\bJ_B^{\rm H}\tilde{\bQ}_B^{-1})$ 
and where $\bJ_B=\tilde{\bP}(\bI+\tilde{\bP}^{\rm H}\tilde{\bP})^{-1}$. 
In the following, we express $g_B(\bQ)$ to a more tractable form. Let  
$\bA_1=\tilde{\bQ}_B^{-1}\bJ_B\bL^{\frac{1}{2}}, \bA_2=\bH_{IB}^{\rm H}\bH_{IB}, \bA_3=\bL^{\frac{1}{2}}\bJ_B^{\rm H}, \bA_4=\bH_{IB}^{\rm H}\tilde{\bQ}_B^{-1}\bJ_B\bL^{\frac{1}{2}}$,
by applying the lemma of matrix identity in \cite{Zhang-17} that for any matrix $\bA$, $\bB$ and diagonal matrix $\bV$ with proper sizes, ${\rm tr}(\bV^{\rm H}\bA\bV\bB)=\bv^{\rm H}(\bA\odot\bB^{\rm T})\bv$ holds where the entries in $\bv$ are all diagonal elements in $\bV$, $g_B(\bQ)$ can be further expressed as
\bal
\notag
g_B(\bQ)&=-{\rm tr}(\bQ^{\rm H}\bA_2\bQ\bA_3\bA_1)+{\rm tr}(\bA_4\bQ^{\rm H})+{\rm tr}(\bQ\bA_4^{\rm H})\\
\label{second2}
&=-g_b(\bq)+2{\rm Re}\{\bq^{\rm H}\ba_4\}
\eal
where $g_b(\bq)=\bq^{\rm H}\bZ\bq$, $\bZ=\bA_2\odot(\bA_3\bA_1)^{\rm T}$ 
and where the entries in $\ba_4$ are all diagonal entires in $\bA_4$.
Hence, P13 can be further approximated to ${\rm P14}$.
\bal
\notag
{\rm P14}: &\underset{\bq}{\max}\ -g_b(\bq)+2{\rm Re}\{\bq^{\rm H}\ba_4\}+\sum_{j=1}^{2}C_j(\tilde{\bQ}),
\\ 
\notag
&s.t.\  |q_i|=1, \forall i.
\eal
We note that the objective function in P14 is quadratic concave in $\bq$, however, it is still difficult to solve this problem due to non-convex UMC. Hence, we apply the following lemma, from which the proof can be found in \cite{Sun-17}.
\begin{lemma}
Let $\bX$ be an $n \times n$ Hermitian matrix, then for any point $\tilde{\ba}\in\mathbb{C}^{n\times 1}$, $\ba^{\rm H}\bX\ba$ is upper bounded by
$\ba^{\rm H}\bX\ba\leq\ba^{\rm H}\bY\ba-2Re\{\ba^{\rm H}(\bY-\bX)\tilde{\ba}\}+\tilde{\ba}^{\rm H}(\bY-\bX)\tilde{\ba}$, 
where $\bY=\lambda_1(\bX)\bI$.
\end{lemma}
Using this lemma, given a feasible point $\tilde{\bq}$, a  surrogate function of $g_b(\bq)$ can be expressed as 
\bal
\notag
g_b(\bq) &\leq \bq^{\rm H}\gl_1(\bZ)\bI\bq-2{\rm Re}\{\bq^{\rm H}(\gl_1(\bZ)\bI-\bZ)\tilde{\bq}\}\\
\notag
&\quad+\tilde{\bq}^{\rm H}(\gl_1(\bZ)\bI-\bZ)\tilde{\bq}\\
\label{third}
&=2n\gl_1(\bZ)-2{\rm Re}\{\bq^{\rm H}(\gl_1(\bZ)\bI-\bZ)\tilde{\bq}\}-\tilde{\bq}^{\rm H}\bZ\tilde{\bq}
\eal
where $\tilde{\bq}$ is the feasible point, the entries of which are the diagonal entries of $\tilde{\bQ}$. Hence, combing \eqref{second}, \eqref{second2} and \eqref{third}, $g(\bQ)\geq \tilde{g}(\bQ,\tilde{\bQ})$ follows.
It can be also verified that $\tilde{g}(\bQ,\tilde{\bQ})$ is a surrogate function of $g(\bQ)$.
By dropping the constant term in $\tilde{g}(\bQ,\tilde{\bQ})$, ${\rm P14}$ is finally approximated to 
$\max_{\bq} {\rm Re}\{\bq^{\rm H}\bv\},  s.t.\ |\bq_i|=1$, 
where $\bv=(\gl_1(\bZ)\bI-\bZ)\tilde{\bq}+\ba_4$. Obviously, the objective function ${\rm Re}\{\bq^{\rm H}\bv\}$ is maximized only when the phase of $\bq$ and $\bv$ are equal. Thus, the closed-form global optimal solution for P15 is expressed as \eqref{optimal_q} 
where $v_i$ denotes the $i$-th elements of $\bv$, from which the proof is complete.



\begin{thebibliography}{99}


\bibitem{Bloch-11} M. Bloch and J. Barros, ``Physical-layer security: from information theory to security engineering," \emph{Cambridge University
Press}, 2011.



\bibitem{Khisti-1} A. Khisti and G.W. Wornell, ``Secure transmission with multiple antennas -- part I: The MISOME wiretap channel," \emph{IEEE Trans. Inf. Theory}, vol. 56, no. 7, Jul. 2010.

\bibitem{Li-13} Q. Li and W. K. Ma, ``Spatially selective artificial-noise aided transmit optimization for MISO Multi-Eves secrecy rate maximization,"  \emph{IEEE Trans. Signal Process.}, vol. 61, no. 10, pp. 2704--2717, May  2013.

\bibitem{Khisti-2} A. Khisti and G.W. Wornell, ``Secure transmission with multiple antennas---part II: The MIMOME wiretap channel," \emph{IEEE Trans. Inf. Theory}, vol. 56, no. 11, pp. 5515-5532, Nov. 2010.









\bibitem{Loyka-15} S. Loyka and C. D. Charalambous,  ``An algorithm for global maximization of secrecy rates in Gaussian MIMO wiretap channels," \emph{IEEE Trans. Commun.}, vol. 63, no. 6, pp. 2288-2299, June. 2015.

\bibitem{Dong-18} L. Dong, S. Loyka, and Y. Li,  ``The secrecy capacity of Gaussian MIMO wiretap channels under interference constraints,"  \emph{IEEE  J. Sel. Areas Commun.}, vol. 36, no. 4, pp. 704-722, Apr. 2018.

\bibitem{Loyka-18} S.Loyka and L. Dong, ``Optimal full-rank signaling over MIMO wiretap
channels under interference constraint," \emph{IEEE Wireless Commun. Letters}, vol. 7, no. 4, pp. 534-537, Aug. 2018.

\bibitem{Wu-18c} Y. Wu \emph{et al.}, ``A survey of physical layer security techniques for 5G
wireless networks and challenges ,"  \emph{IEEE J. Sel. Areas Commun.}, vol. 36, no. 4,
pp. 679-695, Apr. 2018.

\bibitem{Zhao-19} J. Zhao, ``A survey of intelligent reflecting surfaces (IRSs):
Towards 6G wireless communication networks
with massive MIMO 2.0,"  2019, \emph{arXiv:1907.04789}. [Online]. Available: https://arxiv.org/pdf/1907.04789



\bibitem {Hu-18} S. Hu, F. Rusek, and O. Edfor, ``Beyond massive MIMO: The potential of data transmission with large intelligent surfaces," \emph{IEEE Trans  Signal Process.}, vol. 66, no. 10, pp. 2746-2758, Mar. 2018.

\bibitem {Ntontin-19} K. Ntontin, \emph{et al.},  ``Reconfigurable intelligent surfaces vs. relaying: Differences, similarities, and performance comparison," \emph{IEEE Open J. Commun. Soc.}, vol. 1, pp. 798-807, Jul. 2020.



\bibitem{Wu-19} Q. Wu and R. Zhang, ``Towards smart and reconfigurable environment: Intelligent reflecting surface aided wireless network," \emph{IEEE Commun. Magazine}, vol. 58, no. 1, pp. 106-112, Jan. 2020.




\bibitem {Wu-19b}  Q. Wu and R. Zhang, ``Intelligent reflecting surface enhanced wireless
network via Joint active and passive beamforming,"  \emph{IEEE Trans. Wireless Commun.},vol. 18, no. 11, pp. 5394-5409, Nov. 2019.


\bibitem {Wu-18} Q. Wu and R. Zhang,  ``Beamforming optimization for intelligent reflecting surface with discrete phase shifts."  \emph{2019 IEEE International Conference on Acoustics, Speech and Signal Processing (ICASSP)}, Brighton, UK, 2019.


\bibitem {Huang-18} C. Huang, A. Zappone, M. Debbah, and C. Yuen, ``Achievable rate maximization by passive intelligent mirrors," \emph{IEEE International Conference on Acoustics, Speech and Signal Processing (ICASSP)}, Calgary, Canada, Apr. 2018.


\bibitem {Huang-19} C. Huang, A. Zappone, G. C. Alexandropoulos, M. Debbah, and C. Yuen,  ``Reconfigurable intelligent surfaces for energy efficiency in wireless communication," \emph{IEEE Trans. Wireless Commun.}, vol. 18, no. 8, pp. 4157-4170, Aug. 2019.

\bibitem {Huang-18b} C. Huang, G C. Alexandropoulos, A. Zappone, M. Debbah, and C. Yuen,  ``Energy efficient multi-user MISO communication using low resolution large intelligent surfaces," \emph{IEEE Globecom Workshops (GC WKshps)}, Abu Dhabi, United Arab Emirates, 2018.

\bibitem {Zhang-19} S. Zhang and R. Zhang, ``Capacity characterization for intelligent reflecting surface aided MIMO communication,"  \emph{IEEE J. Sel. Commun.}, to be published,  DOI: 10.1109/JSAC.2020.3000814.






\bibitem {Shen-19} H. Shen, W. Xu, S. Gong, Z. He, and C. Zhao, ``Secrecy rate maximization for intelligent reflecting surface assisted multi-antenna communications,"  \emph{IEEE Commun.  Letters}, vol. 23, no. 9, pp. 1488-1492, Sep. 2019.

\bibitem {Cui-19} M. Cui, G. Zhang, and R. Zhang, ``Secure wireless communication via intelligent reflecting surface,"  \emph{IEEE Wireless Commun. Letters}, vol. 8,
no. 5, pp. 1410-1414, Oct. 2019.




\bibitem{Song-20} Y. Song, M. R. A. Khandaker, F. Tariq, and K.-K. Wong, ``Truly intelligent reflecting surface-aided secure communication using deep learning," 2019, \emph{arXiv:2004.03056}. [Online]. Available: https://arxiv.org/abs/2004.03056.

\bibitem{Zheng-20} C. Zheng, W. Hao, P. Xiao, and J. Shi, ``Intelligent reflecting surface aided multi-antenna secure transmission,"  \emph{IEEE Wireless Commun. Letters}, vol. 9,
no. 1, pp. 108-112, Jan. 2020.

\bibitem {Yu-19} X. Yu and R. Schober, ``Enabling secure wireless communications via intelligent reflecting surfaces," in \emph{Proc. IEEE Global Commun. Conf. (GLOBECOM)}, Waikoloa, HI, USA, Dec. 2019, pp. 1-6.




\bibitem {Feng-19}  B. Feng, Y. Wu, and M. Zheng,  ``Secure transmission strategy for intelligent reflecting surface enhanced wireless system," \emph{2019 11th International Conference on Wireless Communications and Signal Processing (WCSP)}, Xi'an, China, 2019.



\bibitem  {Chen-19} J.  Chen ,Y.  Liang , Y. Pei, and H. Guo, ``Intelligent reflecting surface: A programmable wireless environment for physical layer security," \emph{IEEE Access}, vol. 7, pp. 82599-82612, Jun. 2019.


\bibitem {Xu-19} D. Xu \emph{et al.}, ``Resource allocation for secure IRS-assisted multiuser MISO systems,"  \emph{2019 IEEE Globecom Workshops (GC WKshps)}, Waikoloa, HI, USA, Dec. 2019.

\bibitem {Dong-20} L. Dong, H.-M. Wang, ``Secure MIMO transmission via intelligent reflecting surface," \emph{IEEE Wireless Commun. Letters}, vol. 9, no. 6, pp. 787-790, Jun. 2020.


\bibitem {Wang-19} Z. Wang, L. Liu, and S. Cui, `` Channel estimation for intelligent reflecting surface
assisted multiuser communications," \emph{2020 IEEE Wireless Communications and Networking Conference (WCNC)}, Seoul, Korea (South), May 2020.


\bibitem {He-19} Z.-Q. He and X. Yuan, ``Cascaded channel estimation for large intelligent
metasurface assisted massive MIMO,"   \emph{IEEE Wireless Commun. Letters}, vol. 9, no. 2, pp. 210-214, Feb. 2020.




\bibitem {Mirza-19} J. Mirza and B. Ali, ``Channel Estimation Method and Phase Shift Design for Reconfigurable Intelligent Surface Assisted MIMO Networks," 2019,   \emph{arXiv:1912.10671}. [Online]. Available: https://arxiv.org/abs/1912.10671


\bibitem {Hjorungnes-11} A. Hjorungnes, ``Complex-valued matrix derivatives: With applications in signal processing and communications," \emph{Cambridge University Press},  2011.

\bibitem{Boyd-04} S. Boyd and L. Vandenberghe, ``Convex Optimization," \emph{Cambridge University Press},  2004.

\bibitem{Zhang-99} F. Zhang, ``Matrix theory: Basic results and techniques," \emph{Springer}, 1999.

\bibitem{Wang-12} H.-M. Wang, Q. Yin, and X. Xia, ``Distributed beamforming for physical-layer security of two-way relay networks," \emph{IEEE Trans.  Signal Process.}, vol. 60, no. 7, pp. 3532-3545, Jul. 2012.

\bibitem{Sun-17} Y. Sun, P. Babu, and D. P. Palomar,   ``Majorization-minimization algorithms in signal
processing, communications, and machine learning," \emph{IEEE Trans. Signal Process.}, vol. 65, no. 3, pp. 794-816, Feb. 2017.


\bibitem{Zhang-17} X.-D. Zhang, ``Matrix analysis and applications," \emph{Cambridge University Press},  2017.


\bibitem{Fang-15}B. Fang, Z. Qian, W. Zhong, and W. Shao, ``AN-aided secrecy precoding for SWIPT in cognitive MIMO broadcast channels,"
\emph{IEEE Commun. Letters}, vol. 19, no. 9, pp. 1632-1635, Sep. 2015.


\end{thebibliography}
\end{document}